\documentclass[11pt,a4paper]{article}
\usepackage[american]{babel}
\usepackage{amsmath,amssymb,amsthm}
\usepackage{newtxtext,newtxmath}
\usepackage{microtype}
\usepackage{graphicx}
\usepackage{booktabs}
\usepackage{longtable}
\usepackage{array}
\usepackage{calc}
\usepackage{etoolbox}
\usepackage{url}
\usepackage{cmap}
\usepackage{caption}
\newlength{\TabCapGap}
\newcommand{\TABCAP}[1]{\par{\centering\small\textbf{#1}\par}\vspace*{-\TabCapGap}}
\usepackage{xcolor}
\usepackage{needspace}
\usepackage{placeins}
\usepackage[numbers,sort&compress]{natbib}
\usepackage{geometry}
\AtBeginEnvironment{longtable}{\scriptsize\setlength{\tabcolsep}{4.5pt}}
\AtBeginEnvironment{tabular}{\scriptsize\setlength{\tabcolsep}{4.5pt}}

\graphicspath{{figures/}}

\usepackage{algorithm}
\floatname{algorithm}{Algorithm}
\usepackage{algpseudocode}

\newtheoremstyle{encompact}{3pt}{3pt}{\itshape}{}{\bfseries}{.}{0.5em}{}
\theoremstyle{encompact}
\newtheorem{definition}{Definition}
\newtheorem{problem}{Problem}
\newtheorem{proposition}{Proposition}
\newtheorem{assumption}{Assumption}
\AtBeginEnvironment{thebibliography}{\footnotesize\setlength{\bibsep}{0pt plus 0.2ex}}

\title{Ineffective Occupancy in Heterogeneous GPU Clusters: Prediction, Contribution Analysis, and Readiness-Aware Scheduling}
\author{%
Tao Zhang\textsuperscript{1}\quad Bin Liao\textsuperscript{2,3,*}\\[6pt]
{\small 1. School of Information Engineering, Guizhou University of Traditional Chinese Medicine, Guiyang 550025, Guizhou, China}\\
{\small 2. School of Big Data Statistics, Guizhou University of Finance and Economics, Guiyang 550025, Guizhou, China}\\
{\small 3. School of Mathematics and Statistics, Guizhou University of Finance and Economics, Guiyang 550025, Guizhou, China}\\[4pt]
{\small (zhangtao031@gzy.edu.cn)}}
\date{}

\usepackage[hidelinks,unicode]{hyperref}
\hypersetup{
  pdftitle={Ineffective Occupancy in Heterogeneous GPU Clusters: Prediction, Contribution Analysis, and Readiness-Aware Scheduling},
  pdfkeywords={GPU clusters, ineffective occupancy, readiness prediction, admission ordering},
  bookmarksnumbered=true,
}

\makeatletter
\patchcmd{\@maketitle}{\null\vskip 2em}{\null\vskip 0.5em}{}{%
  \PackageWarning{paper-p05-en-v32}{A2 title-top patch FAILED}}
\makeatother

\begin{document}
\maketitle
\begingroup
\renewcommand{\thefootnote}{}\footnotetext{Funding: Guizhou Provincial Basic Research Program (Natural Science) (No. MS[2025]226); Doctoral Startup Fund of Guizhou University of Traditional Chinese Medicine (No. [2024]07); Scientific Research Startup Project for Introduced Talents of Guizhou University of Finance and Economics (No. 2023YJ26). Corresponding author: Bin Liao (E-mail: liaobin665@163.com).}
\endgroup

\begin{abstract}
In heterogeneous GPU clusters, GPUs are costly and scarce, yet placing a job on a machine does not mean that the job becomes ready to run. A large amount of allocated capacity can remain occupied for long periods without producing useful computation, further aggravating the waste of compute resources. On the Alibaba GPU cluster trace we find that, even when GPU capacity is sufficient and placement is complete, 52.70\% of execution records never became ready. We term this phenomenon ineffective occupancy. Building on this diagnosis, and keeping the placement policy fixed, we propose Readiness-Aware Scheduling (RAS): it uses only nine attributes observable at admission time to predict readiness probability and reorders admission accordingly. Under a five-fold split grouped by machine, the gradient-boosted model reaches an ROC-AUC of 0.8187. A feature-group contribution analysis shows that the model's score is driven mainly by job profile, duration, and GPU model, with the request-size group contributing least, followed by the scheduling and concurrency group. In GPU terms, RAS reduces the ineffective-occupancy ratio (the share of cluster GPUs held under ineffective occupancy) from 43.64\% to 27.15\% and raises the effective coverage of large requests (eight GPUs or more) from 46.11\% to 63.18\%, without sacrificing capacity utilization. Relative to the upper bound obtained by ordering with ground-truth labels, RAS captures 65.1\% of the attainable improvement, well above the interpretable rules built only from the contribution-analysis conclusions; the gain therefore comes mainly from predicting with all nine attributes jointly rather than from combining rules built on any single attribute group. These results indicate that readiness status is a measurable, predictable, and optimizable dimension of admission decisions in heterogeneous GPU clusters.
\end{abstract}


\section{Introduction}

Heterogeneous GPU clusters have become the mainstream infrastructure for large-scale AI training and inference. Unlike homogeneous clusters, heterogeneous clusters deploy multiple GPU models at the same time, and jobs (such as large language model training and online inference services) differ in their requirements for GPU model, GPUs per node, and network topology. This difference changes the basic resource-allocation problem: a homogeneous cluster asks whether any GPU is still available, whereas a heterogeneous cluster must answer whether the available GPUs can satisfy a specific job's requirements---that is, the match between job demand and GPU supply \citep{li2026heterogeneity, weng2022mlaas}.

In practice, however, this match works poorly. A Cast AI survey of more than 23,000 Kubernetes clusters reports an average GPU utilization of only 5\% in enterprise environments \citep{castai2026gpu}; an ACM Computing Surveys review likewise reports that GPU utilization in production clusters at Alibaba, SenseTime, Microsoft, and other organizations typically ranges from 25\% to 50\% \citep{huang2025efficiency}.

The problem is not limited to a mismatch between supply and demand. Even when GPUs have been allocated, a substantial fraction of jobs never become ready. The Cast AI report lists ``allocated but persistently near-zero utilization'' as one of the main types of GPU waste \citep{castai2026gpu}; Alibaba's production traces record the same phenomenon \citep{li2026heterogeneity}. Taken together, these two lines of evidence show that successful allocation does not imply that a job becomes ready: the problem extends from ``GPUs cannot be allocated'' to ``GPUs have been allocated but the job never becomes ready.'' The coexistence of compute scarcity and compute waste has become a prominent obstacle to improving the efficiency of AI infrastructure.

Characterization studies of production clusters further reveal how this contradiction manifests. A study of Alibaba production traces \citep{li2026heterogeneity} points out that high GPU demand does not translate into high effective utilization: idle GPUs cannot be allocated promptly because of missing matched CPUs, network-locality constraints, user over-reservation, and whole GPUs idling across machines. Another characterization study of MLaaS clusters \citep{weng2022mlaas} likewise documents systematic mismatches between the form of GPU supply and job demand. These works, however, focus mainly on cluster-level allocation and fragmentation rates; we focus on a more basic question: can a placed job actually become ready?

We answer this question at the job level. On our trace, placement is usually not the bottleneck: only one record requests more GPUs than the machine it was placed on can hold (Section~3.3). The real loss is not ``cannot fit'' but the substantial share of placed jobs that never become ready. We call this phenomenon ineffective occupancy: GPUs are occupied by a job, but the job never truly runs. Among the 40,129,894 labeled execution records we use, this phenomenon accounts for 52.70\%. Its existence implies that scheduling optimization should not pursue placement counts alone; it should also consider whether a placed job can actually run.

To address this problem, we propose Readiness-Aware Scheduling (RAS). Our contributions are threefold:

\textbf{(1) Identifying and measuring ineffective occupancy:} We focus on the phenomenon that a substantial share of placed jobs never become ready, and we measure its scale and distribution at the job level.

\textbf{(2) Building a readiness predictor and analyzing feature contributions:} Using only nine attributes observable at admission time, we build a readiness predictor and verify its ranking ability under a machine-grouped split. A contribution analysis of feature groups shows that the model's score is driven mainly by job profile, duration, and GPU model rather than being uniformly distributed.

\textbf{(3) Proposing and validating a readiness-aware scheduling algorithm:} With the placement rule fixed, we propose an admission-ordering algorithm based on predicted readiness probability and verify its effect in suppressing ineffective occupancy.

The remainder of this paper is organized as follows. Section~2 reviews related work, Section~3 defines the problem and the data, Section~4 presents the diagnostic analysis, Section~5 develops and evaluates readiness-aware scheduling, and Section~6 discusses the results and concludes.

\section{Related Work}

This section reviews existing work along three lines: cluster characterization and effective utilization, placement optimization and admission ordering, and the mapping from prediction to scheduling. Section~2.4 further contrasts our work with these lines.

\subsection{Cluster Characterization and Effective Utilization}

Cluster characterization studies focus on GPU utilization, resource allocation, and fragmentation \citep{li2026heterogeneity, weng2022mlaas, guo2019alibaba}. Early analyses of Alibaba production traces \citep{guo2019alibaba}, together with measurements of heterogeneous MLaaS clusters \citep{weng2022mlaas}, document a persistent gap between requested size and actual usage. A study of Alibaba production traces \citep{li2026heterogeneity} further attributes this gap mainly to supply-side factors: idle GPUs are not scarce in aggregate but are scattered across machines, lack matched CPUs, violate network-locality constraints, or are held by user reservations; the same study notes that fractional-GPU fragmentation has become relatively minor, while the idling of whole GPUs across machines has become the main problem. Industry reports and reviews estimate utilization at the 5\% and 25\%--50\% levels, respectively \citep{castai2026gpu, huang2025efficiency}; scheduling studies give broad estimates of 5\%--30\% \citep{kunkerkar2026queuetheoretic} or below 50\% \citep{mamirov2025reducing}.

These works take the cluster as the unit of analysis and focus on allocation rate, fragmentation, or overall utilization. For example, the defragmentation algorithm of the Alibaba trace study \citep{li2026heterogeneity} reduces the number of machines with idle resources. But a higher cluster-level allocation rate does not necessarily imply higher effective output at the job level. Work on ghost occupancy \citep{ghost2026} observes a related phenomenon: Kubernetes clusters contain GPUs that are allocated yet produce no useful work, and such allocations are detected and reclaimed. Overall, these works lay the foundation for understanding GPU utilization, but their unit of analysis remains the cluster and does not extend to the job level.

\subsection{Placement Optimization and Admission Ordering}

\textbf{(1) Placement optimization:} Placement optimization is one of the main lines of GPU cluster scheduling and has evolved from generic bin-packing rules to GPU-specific schedulers. Early one-dimensional bin-packing rules (First-Fit, Best-Fit, Worst-Fit) target packing efficiency and struggle to capture the heterogeneity of GPU jobs. FGD \citep{weng2023fgd} moves the objective from ``can the job be placed'' to ``is the remaining space still usable after placement'': the quality of a placement decision depends not only on whether the current job fits but also on whether the remaining space will remain usable in the future.

Later GPU-specific schedulers incorporate job characteristics: Tiresias \citep{gu2019tiresias} orders jobs by least attained service; Gavel \citep{narayanan2020gavel} solves a linear program over a speed matrix; Pollux \citep{qiao2021pollux} chooses the number of GPUs by goodput; Sia \citep{jayaram2023sia} mixes coarse- and fine-grained resource allocation; and Synergy \citep{lao2022synergy} allocates according to resource sensitivity. These works cover the two core questions of placement: where to place a job and how many GPUs to allocate to it. Industrial schedulers typically abstract placement as a pluggable scoring module; for example, Borg (Google's cluster management system) \citep{verma2015borg} and Kubernetes \citep{k8s2024scheduler} implement placement rules as configurable scoring items.

\textbf{(2) Admission ordering:} Admission ordering has recently attracted attention as a problem in its own right. Work in queueing theory \citep{kunkerkar2026queuetheoretic} proves that under multi-dimensional resource demands the optimal admission-ordering problem is NP-hard, and points out that the greedy policy of Kueue (a Kubernetes job queueing controller) lacks theoretical guarantees in the waiting-time dimension. In contrast to this theoretical analysis, industry has made admission a separate controller: Runway (an AI video company) uses Kueue as its admission controller and raises GPU utilization by more than 20 percentage points \citep{runway2026nogpuidle}. In that architecture, Kueue decides whether a job starts, and kube-scheduler (the default Kubernetes scheduler) \citep{k8s2024scheduler} decides where admitted pods run. Another line starts from scheduling as a whole: a multi-cluster compute scheduling algorithm based on deep deterministic policy gradient (DDPG) reinforcement learning \citep{hu2025ddpg} solves scheduling as an integrated optimization problem and does not treat admission ordering separately.

\subsection{From Prediction to Scheduling}

By prediction target and scheduling interface, existing work falls into three categories---duration prediction, resource prediction, and readiness/startup-delay prediction; learning-based scheduling can serve as a separate decision interface.

\textbf{(1) Duration prediction and learning-based scheduling:} Duration prediction is the most mature direction: it answers ``how long will the job run'', outputs a scalar duration, and serves placement or backfilling decisions \citep{liu2025ora}. Another line is learning-based scheduling: replacing hand-written heuristics with policy networks and embedding predictors into the decision policy \citep{mao2019decima}. GPARS \citep{wang2024gpars} uses a graph attention network to predict job duration and reduces waiting time by 10.29\% and improves resource utilization by 7.47\% on the Alibaba and Philly traces; its prediction target is still duration, and the prediction is used for GPU-type assignment.

\textbf{(2) Resource prediction:} On GPU clusters, prediction-driven scheduling has made several advances in recent years. PRISM \citep{yuan2025prism} uses an LSTM to predict job arrival patterns and resource demands, reducing mean job completion time (JCT) by 25\%--40\%. Wind \citep{lu2026wind} uses historical execution patterns to predict task demands and completion times and outperforms the default Kubernetes scheduler on a 139-machine cluster. GFS \citep{duan2025gfs} uses a lightweight prediction model to anticipate GPU demand across tenants, reducing the eviction rate by 33.0\% and queueing delay by 44.1\%.

\textbf{(3) Readiness prediction and startup-delay estimation:} A study of the Alibaba GPU trace \citep{tu2024llm} predicts instance readiness status, scheduling delay, and capacity risk for a recommendation-service scenario, using a large language model (LLM) augmentation approach. On container startup delay, ADAPT \citep{baghel2026adapt} tracks container cold-start duration with an online exponentially weighted moving average (EWMA) estimator and feeds it as a dynamic variable into a model-predictive controller to optimize the replica count and planning horizon; CLUEAI \citep{jayahanthan2025clueai} combines machine learning with statistical methods to predict CPU allocation during container startup, reducing average application startup time by 19.9\%.

\subsection{Differences from Existing Work}

Table~1 compares this paper with four representative lines of work.

\TABCAP{Table 1. Differences between this paper and representative work.}

{\def\LTcaptype{none} 
\begin{longtable}[]{@{}
  >{\raggedright\arraybackslash}p{(\linewidth - 10\tabcolsep) * \real{0.1667}}
  >{\raggedright\arraybackslash}p{(\linewidth - 10\tabcolsep) * \real{0.1667}}
  >{\raggedright\arraybackslash}p{(\linewidth - 10\tabcolsep) * \real{0.1667}}
  >{\raggedright\arraybackslash}p{(\linewidth - 10\tabcolsep) * \real{0.1667}}
  >{\raggedright\arraybackslash}p{(\linewidth - 10\tabcolsep) * \real{0.1667}}
  >{\raggedright\arraybackslash}p{(\linewidth - 10\tabcolsep) * \real{0.1667}}@{}}
\toprule\noalign{}
\begin{minipage}[b]{\linewidth}\raggedright
Aspect
\end{minipage} & \begin{minipage}[b]{\linewidth}\raggedright
Cluster characterization\citep{li2026heterogeneity, weng2022mlaas}
\end{minipage} & \begin{minipage}[b]{\linewidth}\raggedright
Fragmentation scheduling\citep{weng2023fgd}
\end{minipage} & \begin{minipage}[b]{\linewidth}\raggedright
GPU-specific schedulers\citep{gu2019tiresias, narayanan2020gavel, qiao2021pollux, jayaram2023sia}
\end{minipage} & \begin{minipage}[b]{\linewidth}\raggedright
Prediction-driven scheduling\citep{liu2025ora, mao2019decima, yuan2025prism, lu2026wind, duan2025gfs, wang2024gpars, tu2024llm}
\end{minipage} & \begin{minipage}[b]{\linewidth}\raggedright
This paper
\end{minipage} \\
\midrule\noalign{}
\endhead
\bottomrule\noalign{}
\endlastfoot
Research object & Cluster & Placement decision & Placement decision & Job/cluster & Job \\
Decision & Analysis (no decision change) & Placement & Placement and GPU count & Placement/allocation & Admission order (placement fixed) \\
Adjustment & None (analysis) & Placement rule & Placement + GPU count & Placement/allocation & Admission order \\
Uses prediction & No & No & No & Yes (duration/demand) & Yes (readiness probability) \\
Prediction target & --- & --- & --- & Duration, arrival pattern, resource demand & Binary readiness probability \\
Optimization goal & Allocation rate, fragmentation & Fragmentation increment & Goodput, JCT & JCT, utilization & Effective output (share of GPUs actually used) \\
Evaluation & Cluster level (production) & Fragmentation rate (production) & JCT, makespan & JCT, utilization & Job-level ready output \\
\end{longtable}
}

\noindent{\footnotesize Note: the seven aspects describe, for each line, the unit of analysis, the decision content, the adjustment mechanism, whether prediction is used, the prediction target, the optimization goal, and the evaluation method; the ``This paper'' column is the reference. ``Job-level ready output'' means that historical readiness labels serve as the outcome attribute in replay (see Section~5.1). Cold-start prediction work \citep{baghel2026adapt, jayahanthan2025clueai} is related, but its scenario does not map directly onto GPU admission, so it is not listed.}

Taken together, the main differences between this paper and the four lines above can be summarized in two respects.

(1) Different problem framing. Cluster characterization remains at the cluster level and focuses on allocation rate, fragmentation, or overall utilization; we move the unit of analysis down to individual jobs and reveal ineffective occupancy---``GPU capacity is sufficient but the job never becomes ready.'' Accordingly, we propose RAS, which uses predicted readiness probability as the admission score and admits jobs in descending order, thereby reducing the extent to which GPUs are held by jobs that remain unready for long periods.

(2) Different decision variables and prediction targets. Fragmentation scheduling and GPU-specific schedulers optimize the placement choice; prediction-driven scheduling predicts duration, arrival patterns, or resource demand and embeds the predictor in the placement algorithm. We fix the placement rule, adjust only the admission order, and predict the job-level binary readiness probability. Our predictor is therefore not embedded in the placement decision but used directly for admission ordering. The evaluation target differs accordingly: existing work asks whether placement choices are sound and whether utilization improves, whereas we ask whether GPUs are actually put to use.

\section{Problem Definition and Data}

Resource decisions are a counting problem in homogeneous clusters and a matching problem in heterogeneous clusters: the former asks ``is there any GPU left?'', whereas the latter asks ``do the available GPUs satisfy the job's demand?'' This section turns ineffective occupancy into a measurable, interpretable object: we first give the definitions, then state three problems and fix the notation, and then introduce the data and the evaluation protocol.

\subsection{Definition of Ineffective Occupancy}

\begin{definition}[Effective and ineffective occupancy]
A job-matching instance consists of a triple $(q_j,\, C_j,\, y_j)$:
\begin{itemize}
\item $q_j = (s_j, g_j, t_j, \mu_j, p_j, a_j, d_j)$ is the request of job $j$; the seven components are the accelerator model $s_j$, the number of requested GPUs $g_j$, the job type $t_j$, the model type $\mu_j$, the priority $p_j$, the generative-request flag $a_j \in \{0,1\}$, and the declared duration $d_j$;
\item $C_j$ is the cluster state when job $j$ enters admission: the set of machines $\mathcal{M}$, and for each machine $i \in \mathcal{M}$ its model, GPU count $c_i$, CPU capacity, and rack membership;
\item $y_j \in \{0,1\}$ is the readiness label, with $y_j = 1$ indicating that the job becomes ready.\end{itemize}
Readiness means that the job has obtained a GPU allocation and has completed image pulling, container startup, and runtime preparation, so that it can begin effective computation. Suppose job $j$ is placed on machine $i$: if $y_j = 1$, we call the placement effective occupancy; if $y_j = 0$ although the capacity $c_i$ of machine $i$ satisfies $c_i \geq g_j$, we call the placement ineffective occupancy---the job was admitted and the machine on which it was placed had enough GPUs for the request, yet the job never became ready.
\end{definition}

$C_j$ cannot be observed directly. The released data record the machine on which each job actually ran, so what we observe is ``which machine this match landed on'', not ``the state of the whole cluster at admission time''. Definition~1 divides placement outcomes into two classes---effective and ineffective occupancy---and we focus on the latter.

Definition~1 characterizes the occupancy state of a single job; the following metric aggregates it to the cluster level.

\begin{definition}[Ineffective-occupancy ratio]
Given the placed set $P$ of one admission and the total number of cluster GPUs $\Gamma = \sum_{i \in \mathcal{M}} c_i$, the ineffective-occupancy ratio is
\begin{equation}
W(P) \;=\; \frac{1}{\Gamma}\sum_{j \in P}(1 - y_j)\, g_j,
\end{equation}
that is, the share of cluster GPUs held under ineffective occupancy.
\end{definition}

$W(P)$ is the main metric of Section~5: it separates ``GPUs are occupied'' from ``GPUs are used effectively'' and aggregates Definition~1 to the policy level.

We use readiness status as the observable indicator. The official definition in the released data is ``whether the pod has ever been ready''. ASI also provides a ``Standby'' mechanism: high-priority users can voluntarily release allocated GPUs and CPUs while keeping the container, and restore them when needed. Voluntary yielding is recorded by the \texttt{state\_public} field of the pod-hour table, which is not included in the execution-summary table we use; hence ``not ready'' in this paper does not include voluntary yielding. In addition, this convention does not exclude executions that exit shortly after becoming ready; such records count as effective occupancy. On these definitions, the paper addresses three questions. Let $x_j$ denote the feature vector formed by the nine observable attributes (the seven request components of Definition~1 plus two scheduling-side attributes, the scheduling delay and the number of concurrent requests on the same machine).

\begin{problem}[Readiness prediction]
Given the feature vector $x_j$ and the readiness label $y_j$, learn a probabilistic model $\hat f:\ x_j \mapsto \Pr(y_j = 1 \mid x_j)$, and evaluate its ranking ability under a machine-grouped protocol.
\end{problem}

\begin{problem}[Structural contribution analysis]
Given a fitted $\hat f$ and the feature-group partition $\mathcal{G}$ defined in Section~4.1, measure how much each group drives the predicted value of $\hat f$, and check the correspondence of this structure on the supply side and in the spatial distribution.
\end{problem}

\begin{problem}[Readiness-aware scheduling]
Given a set of jobs $\mathcal{J}$, each job's requested GPU count $g_j$ and readiness probability $\hat f(x_j)$, and cluster capacity given by the idle GPU counts $\{u_i\}_{i \in \mathcal{M}}$ of the machines in $\mathcal{M}$, determine the admission order under a fixed placement rule so as to minimize the ineffective-occupancy ratio $W(P)$.
\end{problem}

Problems~1 and 2 are answered in Section~4 and Problem~3 in Section~5. The three connect in sequence: readiness probability provides the object to be explained by the contribution analysis, and the contribution analysis delimits the room for adjustment at admission. Table~2 lists the main notation.

\needspace{15\baselineskip}
\TABCAP{Table 2. Main notation.}

{\centering\small
\begin{tabular}{@{}l@{\quad}p{0.29\textwidth}@{\hspace{1.5em}}l@{\quad}p{0.29\textwidth}@{}}
\toprule
\multicolumn{2}{@{}l}{\bfseries Symbol} & \multicolumn{2}{l@{}}{\bfseries Meaning} \\
\midrule
$j$, $\mathcal{J}$ & Job and job set & $i$, $\mathcal{M}$ & Machine and machine set \\
$q_j$ & Job request (components in Definition~1) & $c_i$ & GPU capacity of machine $i$ \\
$x_j$ & Nine observable attributes & $c_{\max}$, $\Gamma$ & Maximum GPUs per machine and total cluster GPUs \\
$y_j$ & Readiness label, $\{0,1\}$ & $C_j$ & Cluster state at admission \\
$\hat f$ & Readiness probability model & $\pi_j$ & Admission score, $\pi_j = \hat f(x_j)$ \\
$\mathbf{u}$, $u_i$ & Idle-GPU vector and its components & $\rho$ & Load (requested GPUs over total GPUs) \\
$P$ & Placed set of one admission & $W(P)$ & Ineffective-occupancy ratio (Definition~2) \\
$\phi_k(\hat f, x_j)$ & Contribution of feature $k$ & $\bar{\phi}_G$, $\sigma_G$ & Contribution value and contribution share \\
\bottomrule
\end{tabular}\par}

Here $\Gamma$, $\pi_j$, $\mathbf{u}$, $u_i$, $P$, and $W(P)$ are used only in Section~5, and $\phi_k$, $\bar{\phi}_G$, and $\sigma_G$ only in Section~4.3; the remaining notation is used throughout.

\subsection{Dataset and Features}

We use the \texttt{cluster-trace-gpu-v2026} trace (Alibaba Serverless Infrastructure, ASI) publicly released by the Alibaba cluster trace programme \citep{alibaba_clusterdata}. It provides four fact tables: the pod-hour table, the server-hour table, the network-element-hour table, and the job execution-summary table. We use the execution-summary table for demand and the server-hour table for supply; the latter provides each machine's model, GPU count, and CPU capacity, together with its rack (access-switch domain) identifier, and is the only source that links a job request to the machines it could use. Table~3 lists the demand-side data conventions and the nine features, Table~4 gives the machine structure we reconstruct from the supply side, and Appendix~A reports the value ranges and descriptive statistics of the nine features.

\TABCAP{Table 3. Execution-summary table and the nine features used in this paper.}

{\def\LTcaptype{none} 
\begin{longtable}[]{@{}
  >{\raggedright\arraybackslash}p{(\linewidth - 2\tabcolsep) * \real{0.2654}}
  >{\raggedright\arraybackslash}p{(\linewidth - 2\tabcolsep) * \real{0.7346}}@{}}
\toprule\noalign{}
\begin{minipage}[b]{\linewidth}\raggedright
Item
\end{minipage} & \begin{minipage}[b]{\linewidth}\raggedright
Value
\end{minipage} \\
\midrule\noalign{}
\endhead
\bottomrule\noalign{}
\endlastfoot
Records & 40,522,321 \\
Labeled records & 40,129,894 \\
Records with an empty label & 392,427 (0.9684\%) \\
Distinct machine identifiers & 37,701 \\
Accelerator models & 13; job types: 6; priority classes: 3 \\
Features (9) & \texttt{gpu\_spec\_public}, \texttt{priority\_class}, \texttt{job\_type\_public}, \texttt{model\_type\_public}, \texttt{is\_genai\_request}, \texttt{gpu\_request}, \texttt{duration\_hours}, \texttt{schedule\_delay\_sec}, \texttt{server\_req\_count} \\
Label & \texttt{ready\_status} \\
Excluded (observable only after the decision) & \texttt{schedule\_status}, \texttt{ready\_delay\_sec} \\
Aggregation granularity & One record per job execution (only execution spans that request GPUs) \\
\end{longtable}
}

\noindent{\footnotesize Note: the nine features are accelerator model, priority class, job type, model type, generative-request flag, requested GPUs, declared duration, scheduling delay, and the number of concurrent requests on the same machine.}

The properties of the demand side bear directly on the analysis in Section~4. The labels are nearly balanced: 47.30\% of labeled records are ready, so accuracy is uninformative and Section~4.2 therefore relies on threshold-independent ranking metrics. Request sizes are highly concentrated: 78.42\% of records request exactly one accelerator, and only 0.73\% request more than eight. Declared duration is heavy-tailed: the mean is 5.15 hours and the median is 0.42 hours, with the maximum shown in Figure~1. The job-identifier field \texttt{workload\_id} is empty for 89.0\% of labeled records, so it cannot serve as an independent grouping variable.

The structure of the supply side makes matching the core problem. The cluster expanded throughout the observation period: the number of machines grew from 28,036 to 37,707 and the number of GPUs from 111,321 to 155,410. Expansion did not change the distribution of machine structure, however: GPUs per machine is bimodal, the vast majority of machines have only one or two GPUs, and a job requesting eight GPUs can land on only 30.8\% of machines (although those machines hold 74.3\% of the GPUs), while a job requesting sixteen can land on only 7.5\%. The CPU-to-GPU ratio is equally uneven: across model and GPU-count combinations, the number of CPU cores per GPU varies widely (Figure~2).

\begin{figure*}[t]
\centering
\includegraphics[width=\textwidth]{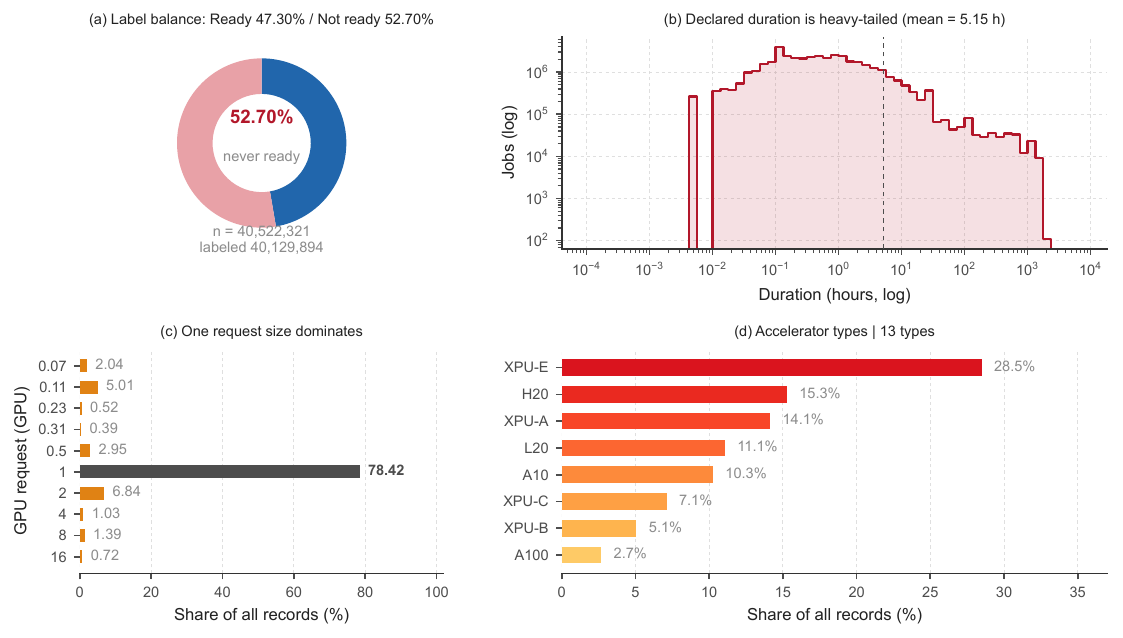}
\caption{(a) Label balance; (b) declared-duration distribution; (c) requested-GPU distribution; (d) share of records by accelerator model.}
\label{fig:1}
\end{figure*}

\TABCAP{Table 4. Supply-side structure reconstructed from the server-hour table (185-day union, 38,048 machines).}

{\def\LTcaptype{none} 
\begin{longtable}[]{@{}
  >{\raggedright\arraybackslash}p{(\linewidth - 2\tabcolsep) * \real{0.4181}}
  >{\raggedright\arraybackslash}p{(\linewidth - 2\tabcolsep) * \real{0.5819}}@{}}
\toprule\noalign{}
\begin{minipage}[b]{\linewidth}\raggedright
Supply attribute
\end{minipage} & \begin{minipage}[b]{\linewidth}\raggedright
Observed value
\end{minipage} \\
\midrule\noalign{}
\endhead
\bottomrule\noalign{}
\endlastfoot
Accelerator models & 13 \\
GPUs per machine & Bimodal, with only a few distinct values; 30.8\% of machines can host a job requesting 8 GPUs \\
CPU cores per GPU & 12--192 \\
Racks & 2,377 \\
Share of racks with a single model & 90.5\% \\
Machines per rack (median/max) & 7 / 136 \\
First day the rack field is available & Day 109 \\
Machines present throughout & 27,573 (of 38,048) \\
Machines never appearing in the job table & 347 \\
Time integrals of demand and supply GPU$\cdot$hours & 254,980,926 / 587,703,816 \\
\end{longtable}
}

Demand GPU$\cdot$hours are about 43.4\% of supply, indicating that capacity is sufficient in aggregate while local matching remains a problem. Two further points about the supply side. First, the rack field is available only in the later part of the observation period, so topology-related analysis covers only that period. Second, the set of machines is identical at different times within the same day---the server-hour table is effectively a daily snapshot---so we use it only at daily granularity.

\clearpage
\begin{figure*}[t]
\centering
\includegraphics[width=\textwidth]{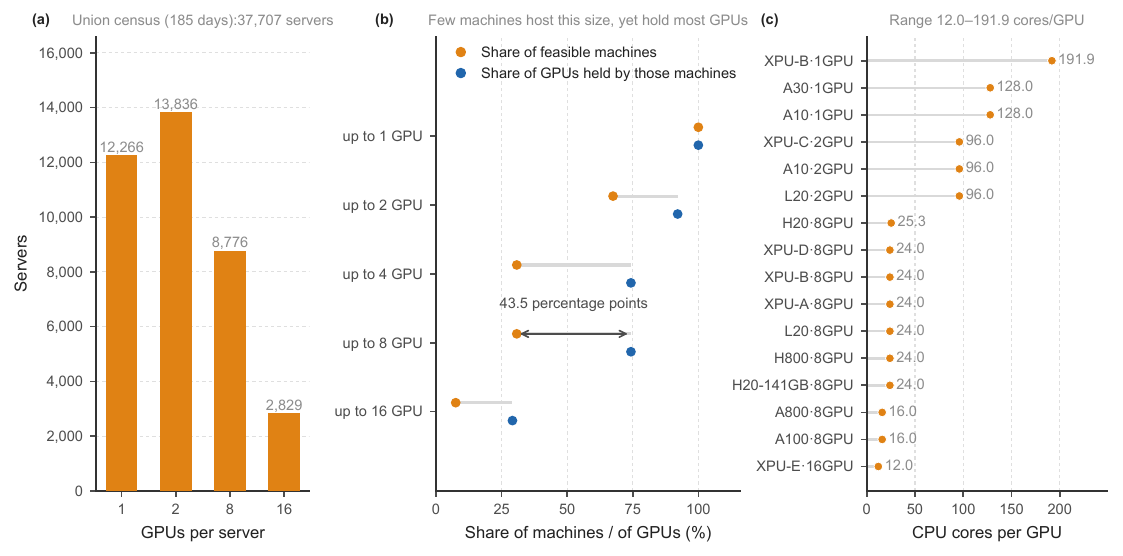}
\caption{(a) Distribution of GPUs per machine; (b) share of machines that can host each request size and the share of GPUs they hold; (c) CPU cores per GPU.}
\label{fig:2}
\end{figure*}

Initial evidence of ineffective occupancy. Applying the ineffective-occupancy case of Definition~1 to the data: of the 40,129,894 labeled records, 52.70\% never became ready, and the GPU capacity of the machines they were placed on is sufficient for the request (test in Section~3.3). This ratio is the core phenomenon of this paper. Section~4 shows that it is predictable and decomposable, and Section~5 gives an admission rule that exploits this prediction.

\subsection{Evaluation Protocol and Leakage Check}

If a field's value is produced only after the admission decision is complete, we exclude it from the feature matrix. All models in Section~4 exclude these fields. The execution-summary table has no time column, so we use a single-point observation and do not analyze temporal drift or seasonality.

\needspace{9\baselineskip}
\begin{samepage}
\TABCAP{Table 5. Leakage check: two fields observable only after the decision and the nine observable attributes.}
\end{samepage}

{\def\LTcaptype{none} 
\begin{longtable}[]{@{}lll@{}}
\toprule\noalign{}
Column & Single-feature ROC-AUC & Treatment \\
\midrule\noalign{}
\endhead
\bottomrule\noalign{}
\endlastfoot
\texttt{schedule\_status} & 0.5055 & Excluded (produced after the decision) \\
\texttt{ready\_delay\_sec} & 0.9644 & Excluded (produced after the decision; directly leaks the label) \\
Nine observable attributes, single-feature ROC-AUC range & 0.5056--0.5786 & Included \\
\end{longtable}
}

The two fields in Table~5 are excluded because of when their values arise: both are produced only after the admission decision is complete and therefore cannot enter the feature matrix. \texttt{schedule\_status} is the scheduler's outcome flag for the execution. Its relation to the label is almost degenerate in one direction: its conditional probability of being 1 is 1.0000 on ready records and 0.9890 on non-ready records; but as a single predictor it carries no information, with a single-feature ROC-AUC of only 0.5055 after aligning its direction with the label. \texttt{ready\_delay\_sec} is the decisive case: it is measured only after a job becomes ready, its single-feature ROC-AUC is 0.9644, and a model that includes it would directly leak the label. By comparison, the single-feature ROC-AUC of the nine observable attributes ranges from 0.5056 to 0.5786.

One further fact concerns the data convention. The execution-summary table contains only execution spans already placed on machines with sufficient GPU capacity: among the 40,129,894 labeled records, only one has \texttt{gpu\_request} exceeding the GPU capacity of the machine it was placed on. Hence the 52.70\% non-ready share is not due to insufficient capacity but to ineffective occupancy. This distinction separates our work from cluster-level utilization analysis; Section~4.4 analyzes its structure.

We split by whole machine: if records from the same machine appeared in both the training and test sets, the model could memorize whether that machine's jobs are usually ready, whereas a new machine at deployment provides no such information. Grouping by machine puts records from the same machine in the same fold, so the evaluation measures generalization to new machines; a random split is optimistic, whereas a machine-grouped split is closer to deployment. We use the machine-grouped split as the primary protocol. Evaluation relies mainly on a threshold-independent ranking metric (ROC-AUC) because the classes are nearly balanced and accuracy is uninformative; intervals come from bootstrap resampling of held-out records.

\FloatBarrier

\section{Diagnosis: Readiness Is Predictable and Decomposable}

This section answers the three questions. The method is given in Sections~4.1 and 4.2, the contribution analysis in Section~4.3, and the structural evidence in Section~4.4. The output is the readiness probability $\hat f(x_j)$ used for admission ordering. Section~4.4 narrows the actionable variable to the admission order and connects to Section~5.

\subsection{Readiness Predictor}

This section answers Problem~1: given the instance set of Definition~1, learn $\hat f:\ x_j \mapsto \Pr(y_j = 1 \mid x_j)$, the probability that a job request enters the ready state. It outputs no scheduling action and is not embedded in a placement algorithm; its only use is to provide an admission score. This differs from the common practice in prediction-driven scheduling, where predictors output duration, arrival patterns, or resource demand and are embedded inside the placement decision; our predictor answers only ``will this job become ready?'', and is used for admission ordering.

The nine observable attributes are partitioned into five disjoint feature groups forming $\mathcal{G} = \{G_{\mathrm{form}}, G_{\mathrm{dur}}, G_{\mathrm{spec}}, G_{\mathrm{req}}, G_{\mathrm{sched}}\}$: the job-profile group $G_{\mathrm{form}}$ (job type $t_j$, model type $\mu_j$, generative-request flag $a_j$, priority $p_j$), the duration group $G_{\mathrm{dur}}$ (declared duration $d_j$), the GPU-model group $G_{\mathrm{spec}}$ (request-side model $s_j$, i.e., \texttt{gpu\_spec\_public}, excluding the placement-machine identifier), the request-size group $G_{\mathrm{req}}$ (requested GPUs $g_j$), and the scheduling-and-concurrency group $G_{\mathrm{sched}}$ (scheduling delay $\ell_j$ and the number of concurrent requests on the same machine $h_j$). The two fields observable only after the decision---the scheduler's decision flag and the ready delay---are always excluded from the feature matrix, because their values arise only after the admission decision and mixing them in would constitute data leakage. This partition is used consistently throughout this chapter.

The data split follows Section~3.3. The six models compared cover the linear, single-tree, and ensemble families; two gradient-boosting implementations are included so that the main result does not depend on a single library \citep{grinsztajn2022why, shwartzziv2022tabular}. Evaluation focuses on ROC-AUC (a measure of ranking ability) because the classes are nearly balanced and accuracy is uninformative; PR-AUC and probability-calibration deviation serve as auxiliary metrics. Intervals come from bootstrap resampling of held-out records, and model differences are tested with a paired test on the same resampling indices.

\subsection{Results}

Under the machine-grouped five-fold split, the models outperform the majority-class and lookup baselines on readiness prediction. The reference model achieves an ROC-AUC of 0.8187 with interval {[}0.8130, 0.8244{]}, versus 0.4960 for the majority-class baseline; on the same held-out predictions, the reference model performs comparably to XGBoost (Table~6). Compared with ``looking up historical readiness rates'', the two lookup methods reach 0.6789 and 0.7520; the reference model's gain over the stronger lookup method (0.7520) is 0.0666 in AUC.

The lookup provides a lower bound under the same information; the model must discriminate further. The lookup is fitted within each training set, so it reads the same information as the model; the only difference is whether the nine attributes are used jointly. A positive gain indicates that the nine attributes jointly discriminate better than any single group, which is why Section~5 uses the joint prediction rather than a single attribute group for admission ordering.

Probability calibration is also usable: the reference model's calibration deviation is 0.0028, so its probability values can be interpreted directly against a threshold; Figure~3 compares the ROC, PR, and calibration curves graphically.

\needspace{9\baselineskip}
\begin{samepage}
\TABCAP{Table 6. Main results: machine-grouped five-fold held-out prediction.}
\end{samepage}

{\def\LTcaptype{none} 
\begin{longtable}[]{@{}lllll@{}}
\toprule\noalign{}
Model & ROC-AUC & 95\% CI & PR-AUC & Calibration deviation \\
\midrule\noalign{}
\endhead
\bottomrule\noalign{}
\endlastfoot
Majority class & 0.4960 & {[}0.4880, 0.5032{]} & 0.4703 & 0.0000 \\
Logistic regression & 0.7028 & {[}0.6962, 0.7096{]} & 0.7065 & 0.0204 \\
Decision tree & 0.7705 & {[}0.7640, 0.7772{]} & 0.7857 & 0.0008 \\
Random forest & 0.8115 & {[}0.8056, 0.8175{]} & 0.8191 & 0.0287 \\
XGBoost & 0.8189 & {[}0.8131, 0.8245{]} & 0.8268 & 0.0016 \\
LightGBM (reference) & 0.8187 & {[}0.8130, 0.8244{]} & 0.8266 & 0.0028 \\
\end{longtable}
}

\noindent{\footnotesize Note: all three metrics are computed from the predictions of the same fit on the same held-out records; intervals are bootstrap 95\% percentile intervals over 20,000 held-out records; the calibration deviation uses ten equal-width bins.}

The values in Tables~6, 8, and 9 come from different fits and are used only for within-table comparison.

\begin{figure*}[t]
\centering
\includegraphics[width=\textwidth]{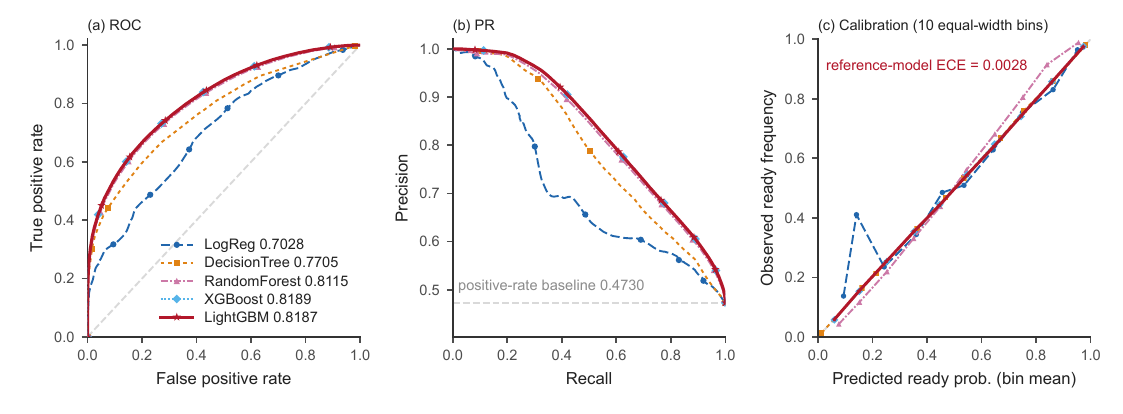}
\caption{(a) ROC curve; (b) PR curve, with the dashed line showing the positive-rate baseline; (c) calibration curve (ten equal-width bins).}
\label{fig:3}
\end{figure*}

\subsection{Contribution of Feature Groups and Ablation Analysis}

The two measures---contribution values and ablation analysis---give somewhat different internal orders but the same conclusion. This section examines the contribution of each of the five feature groups. The global contribution analysis uses path-dependent SHAP (SHapley Additive exPlanations) values, summed by group and averaged over 20,000 held-out records drawn without replacement from the grouped test split; it measures the influence of a group of attributes on the predicted value \citep{lundberg2020treeshap}. Let $\phi_k(\hat f, x_j)$ denote the contribution of feature $k$ on record $x_j$ in sample set $\mathcal{S}$; then the contribution value and contribution share of group $G \in \mathcal{G}$ are

\begin{equation}
\bar{\phi}_G \;=\; \frac{1}{\lvert \mathcal{S} \rvert} \sum_{x_j \in \mathcal{S}} \sum_{k \in G} \bigl\lvert \phi_k(\hat f, x_j) \bigr\rvert,
\qquad
\sigma_G \;=\; \frac{\bar{\phi}_G}{\sum_{G' \in \mathcal{G}} \bar{\phi}_{G'}},
\end{equation}

where $G'$ ranges over the five feature groups in $\mathcal{G}$. Because $\mathcal{G}$ is a disjoint partition, the $\sigma_G$ sum to one, which makes the contribution shares comparable across groups; Eq.~(2) sums absolute values over all features in a group, so the number of features in a group affects the absolute magnitude, and comparisons should be read together with group size.
Table~7 and Figure~4(a) give the contribution value and share of the five groups: job profile 38.3\%, duration 25.6\%, GPU model 18.3\%, together 82.1\%; request size 6.8\%, ranked last; scheduling and concurrency 11.1\%.

\TABCAP{Table 7. Contribution of the five feature groups.}

{\def\LTcaptype{none} 
\begin{longtable}[]{@{}llll@{}}
\toprule\noalign{}
Feature group & Features in group & Contribution value & Share (\%) \\
\midrule\noalign{}
\endhead
\bottomrule\noalign{}
\endlastfoot
Job profile & 4 & 0.9874 & 38.3 \\
Duration & 1 & 0.6606 & 25.6 \\
GPU model & 1 & 0.4710 & 18.3 \\
Scheduling and concurrency & 2 & 0.2857 & 11.1 \\
Request size & 1 & 0.1748 & 6.8 \\
\end{longtable}
}

\noindent{\footnotesize Note: contribution share is the share of the contribution value defined in Eq.~(2) among the five groups, summing to one; rounded values may not sum to exactly 100\%. Contribution values come from the same reference model on 20,000 held-out records; the group partition is unique throughout the paper.}

\begin{figure*}[t]
\centering
\includegraphics[width=\textwidth]{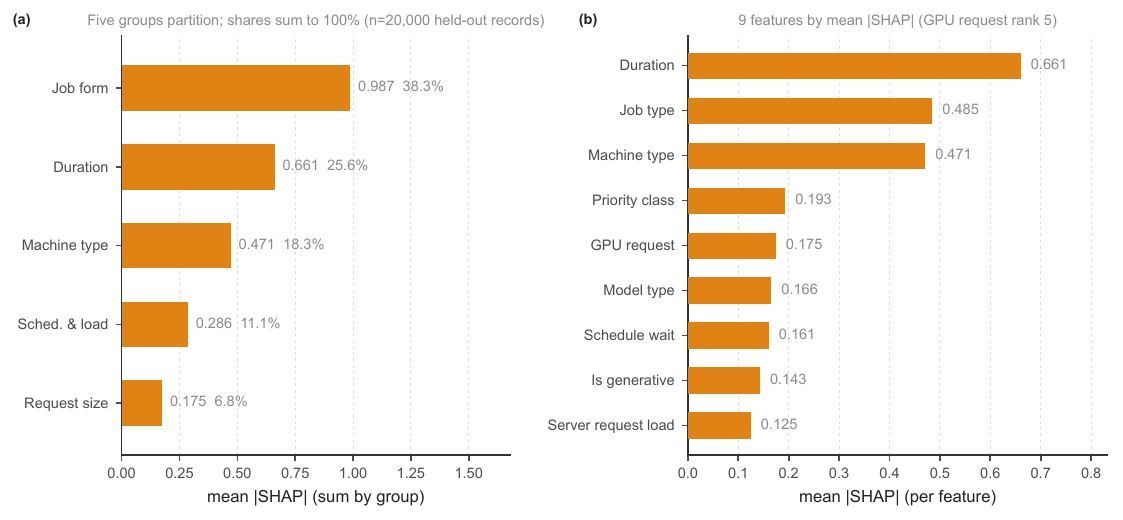}
\caption{(a) Mean absolute Shapley value and contribution share of the five feature groups; (b) mean absolute Shapley value of the nine individual attributes.}
\label{fig:4}
\end{figure*}

Contribution value measures how strongly a group drives the prediction. If one instead uses the performance loss after removing the group, the internal order of the top three changes. This section reports that difference. Table~8 places the two side by side and Figure~5 gives the corresponding bar chart. Path-dependent Shapley values rank job profile first and duration second; ablation analysis ranks duration first. Both are clearly larger than request size and scheduling. The two measures differ because they are different statistics: contribution value is inflated by the number of features in a group, whereas ablation analysis measures whether the group's information can be substituted. The part where the two agree is what the conclusion relies on: both rank request size last (removing it changes AUC by $-0.0098$) and both support the dominance of the three groups; when each group is used alone, duration reaches 0.6567, above job profile (0.6152) and GPU model (0.6222). This section reports the part where they agree: three groups dominate and request size ranks last.

\needspace{9\baselineskip}
\begin{samepage}
\TABCAP{Table 8. Contribution ranking and ablation ranking side by side.}
\end{samepage}

{\def\LTcaptype{none} 
\begin{longtable}[]{@{}lll@{}}
\toprule\noalign{}
Feature group & Contribution value & Performance change after ablation \\
\midrule\noalign{}
\endhead
\bottomrule\noalign{}
\endlastfoot
Duration & 0.6606 & -0.0669 \\
GPU model & 0.4710 & -0.0409 \\
Job profile & 0.9874 & -0.0328 \\
Scheduling and concurrency & 0.2857 & -0.0133 \\
Request size & 0.1748 & -0.0098 \\
\end{longtable}
}

\noindent{\footnotesize Note: the contribution column is the mean absolute value over 20,000 held-out records; the ablation column is the change in test AUC relative to the all-feature model (0.8167); the two columns come from independent fits and are used only to compare rankings.}

\begin{figure*}[t]
\centering
\includegraphics[width=\textwidth]{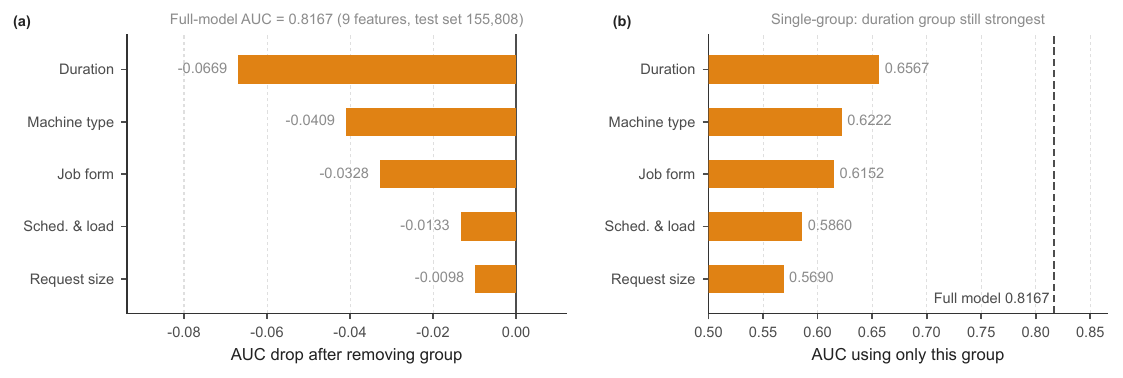}
\caption{(a) Difference in AUC from the all-feature model after removing the group; (b) AUC when trained with only that group.}
\label{fig:5}
\end{figure*}

Figure~6 shows the individual contributions of the nine attributes at the record level; it follows the convention of Figure~4(b) but retains per-record variation.

\begin{figure*}[t]
\centering
\includegraphics[width=\textwidth]{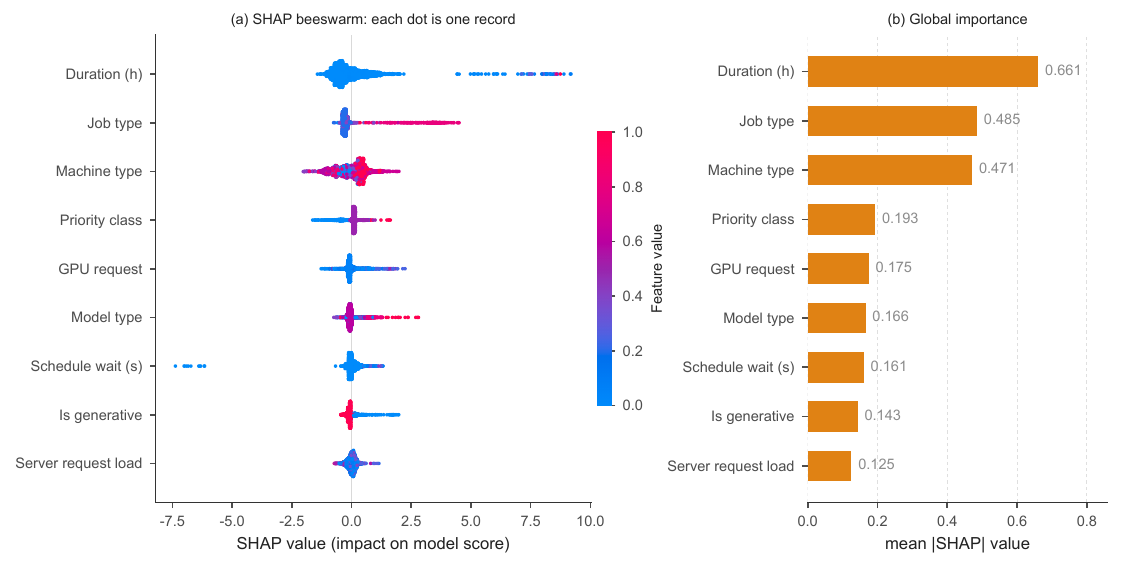}
\caption{(a) SHAP beeswarm for the nine attributes, colored by attribute value; (b) mean absolute SHAP value per attribute.}
\label{fig:6}
\end{figure*}

\subsection{Structural Origin and Concentration}

This section analyzes the supply-side structure and spatial distribution corresponding to the three dominant groups (job profile, duration, GPU model). As noted in Section~3.2, the cluster expanded throughout the observation period but the distribution of machine forms did not change; two points matter here: GPUs per machine is bimodal, and GPU models are not freely interchangeable. Both indicate that ``enough GPUs but the job is not ready'' corresponds to the supply structure rather than to a transient fluctuation on the demand side.

Ineffective occupancy is also spatially concentrated rather than being uniformly random noise across the cluster. Measured by the variance ratio $D$ in Eq.~(3), the observed variance of per-machine readiness rates is far larger than the variance expected from a binomial distribution at the same sample size. The quantile distribution and the values of $D$ are shown in Figure~7. The rack level also shows clear differences.

The quantiles of per-machine readiness and the record-weighted variance ratio use different conventions. The per-machine quantiles include only machines with at least 100 labels and are computed with equal weight per machine (each machine counts once, not weighted by its record count); the overall readiness rate and the variance ratio are record-weighted, and the variance ratio is the observed variance divided by the binomial expected variance:

\begin{equation}
D \;=\; \frac{\sum_{i \in \mathcal{M}^{+}} n_i\,(r_i - \bar{p})^2 \big/ \sum_{i \in \mathcal{M}^{+}} n_i}{\sum_{i \in \mathcal{M}^{+}} n_i \cdot \bar{p}(1-\bar{p})/n_i \big/ \sum_{i \in \mathcal{M}^{+}} n_i},
\end{equation}

where $\mathcal{M}^{+} = \{\, i \in \mathcal{M} : n_i \ge 100 \,\}$ is the set of machines meeting the label threshold, $n_i$ is the number of labeled records on machine $i$, $r_i$ is that machine's readiness rate, and $\bar{p}$ is the record-weighted overall readiness rate. The numerator is the record-weighted observed variance and the denominator is the expected binomial variance at the same sample size. The difference between the two conventions comes from the machine-size distribution: the equal-weight quantile answers ``how low is a typical machine'', whereas the record-weighted ratio answers ``how concentrated is non-readiness in aggregate''; they are not the same quantity and should not substitute for each other.

\begin{figure*}[t]
\centering
\includegraphics[width=\textwidth]{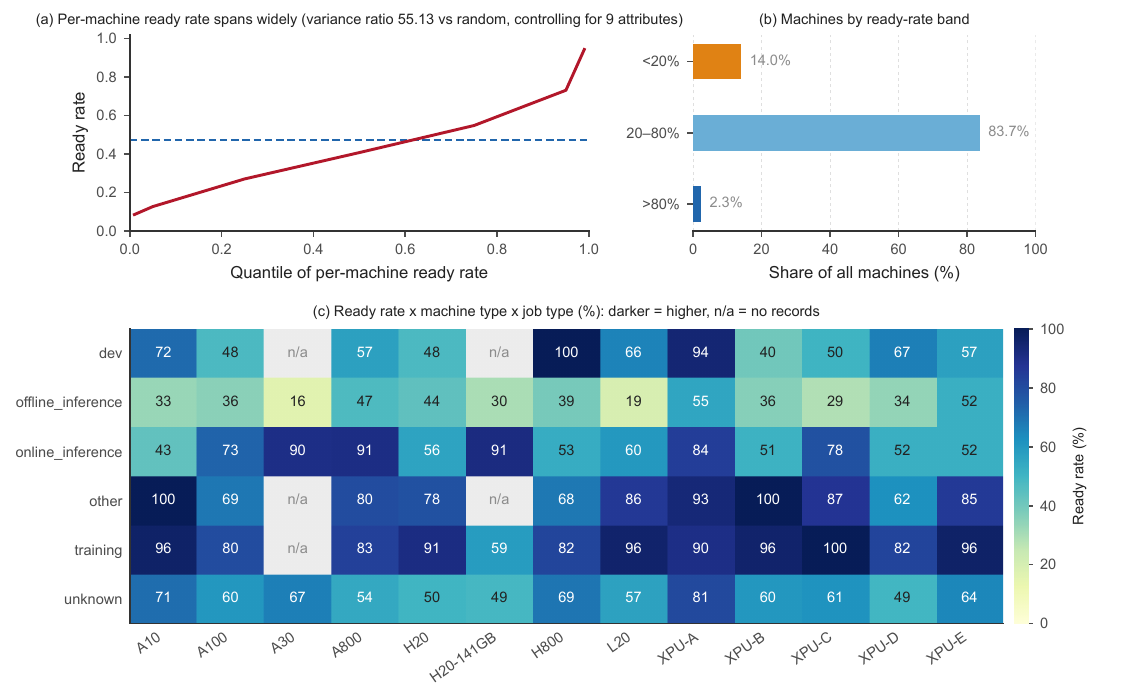}
\caption{(a) Per-machine readiness rates sorted by quantile, with the dashed line showing the overall cluster readiness rate; (b) share of machines by readiness-rate bin; (c) readiness rate across GPU model and job type, with darker color indicating higher readiness and cells with no records marked n/a.}
\label{fig:7}
\end{figure*}

Taken together, these findings narrow the adjustable variables available to Section~5. The model's predicted score is driven mainly by three groups---job profile, duration, and GPU model; the scheduler does not directly change these attributes; the request-size group is not the dominant axis, so policies such as ``large requests first'' do not touch the dominant factors. Combined with the fact noted in Section~3.3 that the execution-summary table contains only execution spans already placed on machines with sufficient GPU capacity (only one record has \texttt{gpu\_request} exceeding the GPU capacity of the machine it was placed on), the 52.70\% non-ready share does not stem from insufficient GPU capacity. Hence the capacity constraint on the placement side is not tight, job attributes cannot be changed, and the only actionable variable left is the admission order. Section~5 therefore connects the predictor to admission and uses the effective output defined in Section~5.1 as the evaluation quantity.

\FloatBarrier

\section{Readiness-Aware Scheduling}

Section~4 narrows the adjustable variable to a single place: the admission order. This section answers Problem~3, turns this restriction into a deployable, comparable rule, and measures it with a different evaluation quantity. Section~5.1 explains why the evaluation quantity must change, Section~5.2 gives the algorithm, complexity, and policy set, Section~5.3 gives the experimental setup and main results, and Section~5.4 performs ablation and sensitivity analyses.

\subsection{From Placement Metrics to Effective Output}

We begin with a minimal example. A cluster has one idle machine and two jobs: one can become ready and one cannot, and the idle GPUs suffice for only one of them. First-come, first-served (FCFS) queues them by arrival; if the job that never becomes ready arrives first, the GPUs are held by a job that cannot become ready. Under a different order, the same GPUs can serve a job that actually runs. In both orders the number of machines and GPUs is unchanged, yet the GPUs actually put to use differ.

This example shows that placement metrics (match rate, allocation rate, fragmentation) measure whether GPUs are occupied and do not answer whether they are used effectively. We focus on the latter: the measurements in Section~4 show that on our trace almost every request that a machine can hold is placed, so placement feasibility is not the bottleneck; the real loss is that a substantial fraction of placed jobs never become ready. The evaluation quantity in this chapter is therefore not how many jobs are placed but the effective output: the share of cluster GPUs occupied effectively, and the job-level metrics derived from it. Definition~2 in Section~3 already defines this quantity as $W(P)$, and this chapter uses it directly as the primary objective; the remaining metrics are computed on the same placed set $P$ under fixed conventions.

To make comparisons fair, this chapter fixes the placement rule of all policies to best-fit (the machine with the smallest idle GPU count that can still hold the job); policies differ only in admission order. The resulting differences therefore come only from admission order, not from the placement algorithm---which is also the only adjustable variable identified by the contribution analysis in Section~4.

\begin{assumption}[Replay assumption]
The admission replay experiment re-simulates admission on historical jobs in a new order, keeps the placement rule fixed, and changes only the admission order; readiness labels are therefore treated as attributes of a job under the current machine supply and enter the comparison.
\end{assumption}

The boundary of this assumption is discussed in Section~6.2.

\subsection{Method: Readiness-Aware Scheduling}

The admission rule can be stated in one sentence: admit in descending order of predicted readiness probability. In deployment it has two stages: offline, a readiness predictor is trained on historically labeled samples; online, its output scores each job and orders admission accordingly. Algorithm~1 gives predictor training, and Algorithm~2 gives the admission procedure; both share the feature-group partition $\mathcal{G}$ and the nine observable attributes from Section~4, and the placement rule stays the same across the two stages.

\begin{algorithm}[!ht]
\caption{Readiness predictor training (RAS stage 1, offline)}
\label{alg:ras_train}
\begin{algorithmic}[1]
\Require Labeled samples $\{(x_j, y_j)\}_{j=1}^{m}$ and grouping key $i_j$ (the machine on which the job actually ran); feature-group partition $\mathcal{G}$; number of folds $K$; number of bins $B$
\Ensure Per-fold reference models $\{\hat f_k\}_{k=1}^{K}$ and held-out readiness probabilities $\{\hat p_j\}_{j=1}^{m}$
\State Build the feature matrix $X \in \mathbb{R}^{m \times 9}$ from the nine observable attributes, excluding the scheduler decision flag and the ready delay, which are observable only after the decision
\State Split the samples into $K$ folds $\{\mathcal{D}_k\}_{k=1}^{K}$ by grouping key $i_j$, so that records from the same machine do not cross folds
\For{$k = 1$ \textbf{to} $K$}
 \State Draw a stratified $10\%$ of the training set $\mathcal{D}_{-k}$ as validation set $\mathcal{D}_k^{\mathrm{val}}$
 \State Fit a gradient-boosted model $\hat f_k$ on $\mathcal{D}_{-k} \setminus \mathcal{D}_k^{\mathrm{val}}$, with early stopping on log loss over $\mathcal{D}_k^{\mathrm{val}}$
 \State Take held-out predictions $\hat p_j \gets \hat f_k(x_j)$ on the held-out set $\mathcal{D}_k$
\EndFor
\State Evaluate the probability-calibration deviation over $B$ equal-width bins and the ranking ability (ROC-AUC, PR-AUC) from the held-out predictions
\State \Return $\{\hat f_k\}$ and $\{\hat p_j\}$
\end{algorithmic}
\end{algorithm}

Algorithm~1 evaluates on held-out predictions: ranking ability by ROC-AUC and PR-AUC, and probability reliability by the expected calibration error (ECE), the arithmetic mean calibration deviation over $B$ equal-width bins,
\begin{equation}
\mathrm{ECE} \;=\; \sum_{b=1}^{B} \frac{\lvert S_b \rvert}{m}\,\bigl\lvert \mathrm{acc}(S_b) - \mathrm{conf}(S_b) \bigr\rvert,
\end{equation}
where $S_b$ is the $b$-th bin and $\mathrm{acc}(S_b)$ and $\mathrm{conf}(S_b)$ are the within-bin label mean and mean predicted probability. The held-out probabilities are not post-processed and are used directly as the readiness probability estimates for all samples.

The ranking key $\pi_j$ in Algorithm~2 takes one of three forms depending on the score source:
\begin{equation}
\pi_j \;=\;
\begin{cases}
-\tau_j, & \text{baseline (arrival order)},\\[2pt]
\hat r(s_j, t_j), & \text{interpretable rule (historical readiness by model $\times$ job type)},\\[2pt]
\hat p_j, & \text{this paper (held-out readiness from Algorithm~1)}.
\end{cases}
\end{equation}
The ranking keys of the remaining baseline policies are given in the ``Basis of admission order'' column of Table~9: shortest-job-first sorts by declared duration $d_j$ ascending, smallest-request-first by $g_j$ ascending, large-request-first by $g_j$ descending, and random ordering uses a uniform random key. All policies share the same placement rule and differ only in the ranking key.

\begin{algorithm}[!ht]
\caption{Readiness-aware admission (RAS stage 2, online)}
\label{alg:ras_admit}
\begin{algorithmic}[1]
\Require Job set $\mathcal{J}$ ($n = \lvert \mathcal{J} \rvert$); requested GPUs $g_j$, arrival order $\tau_j$, and admission score $\pi_j$; machine set $\mathcal{M}$ and idle GPU counts $u_i$; maximum GPUs per machine $c_{\max}$
\Ensure Placement $\alpha_j$ or an unmatched marker for each job, and the ineffective-occupancy ratio $W(P)$
\State Sort $\mathcal{J}$ by $\pi_j$ descending as the primary key and arrival order $\tau_j$ ascending as the secondary key, obtaining admission sequence $\sigma$
\State Maintain bucket counts over the idle-GPU domain $[0, c_{\max}]$ and a query structure for ``the smallest idle count not less than a given value''
\ForAll{each job $j$ in admission sequence $\sigma$}
 \If{$g_j > c_{\max}$}
 \State Mark $j$ as not accommodatable by any machine
 \State \textbf{continue}
 \EndIf
 \State $i^\star \gets \operatorname{BestFit}(g_j)$, the machine with the smallest idle GPU count that is $\ge g_j$
 \If{no such $i^\star$ exists}
 \State Mark $j$ as unmatched
 \State \textbf{continue}
 \EndIf
 \State $u_{i^\star} \gets u_{i^\star} - g_j$, and update the query structure
 \State Record placement $\alpha_j \gets i^\star$, classified as effective or ineffective occupancy according to $y_j$
\EndFor
\State $W(P) \gets \dfrac{1}{\Gamma}\sum_{j \in P}(1 - y_j)\, g_j$
\State \Return $\{\alpha_j\}$, unmatched markers, and $W(P)$
\end{algorithmic}
\end{algorithm}

Proposition~1 gives the computational cost of Algorithm~2; Algorithm~1 is offline training whose cost is dominated by $K$ model fits and does not enter the online admission path.

\begin{proposition}[Complexity]
Let $n = \lvert \mathcal{J} \rvert$. Algorithm~2 runs in $O(n \log n + n \log c_{\max})$, where $c_{\max}$ is the maximum per-machine GPU capacity; in our cluster this value takes only a few distinct values, so the running time is dominated by sorting, $O(n \log n)$.
\end{proposition}

\begin{proof}
The sort on line~1 is $O(n \log n)$. For placement queries, bucket the machines by idle GPU count: the number of buckets is at most $c_{\max} + 1$ (idle counts are integers in $[0, c_{\max}]$), and maintain a segment tree over the GPU-count domain with $O(c_{\max})$ nodes. Querying ``the machine with the smallest idle GPU count not below $g_j$'' and decrementing that machine's idle count each require one root-to-leaf descent in the segment tree, $O(\log c_{\max})$ each; the loop body totals $O(n \log c_{\max})$. Adding the two terms gives the result. In our clusters $c_{\max}$ takes small values (GPUs per machine is bimodal with only a few distinct values), so $\log c_{\max}$ is much smaller than $\log n$ and the running time is dominated by the sorting term.
\end{proof}

Admission ordering relies on the scores from the Section~4 predictor; Figure~8 shows, on two held-out records, how the score is composed from the nine attributes, illustrating that the ordering basis is a decomposable score rather than a threshold on a single attribute.

\begin{figure*}[t]
\centering
\includegraphics[width=\textwidth]{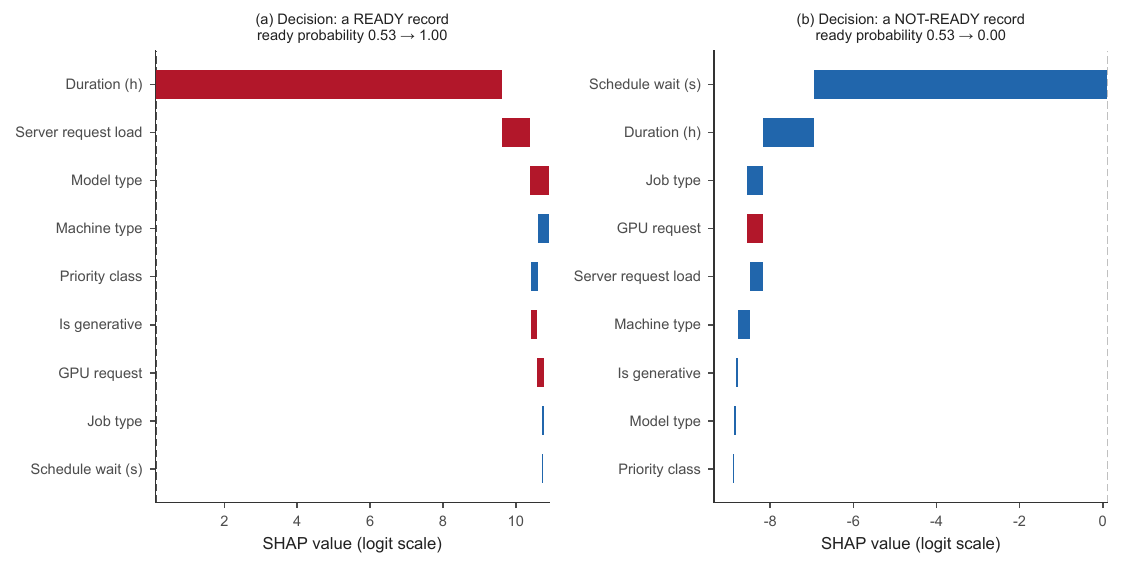}
\caption{(a) and (b) show the SHAP decision process for a ready record and a non-ready record, respectively; red contributions raise the model score and blue ones lower it; the horizontal axis is the SHAP value on the logit scale, and the title gives the predicted readiness of the record.}
\label{fig:8}
\end{figure*}

The score $\pi_j$ has three sources, used to separate whether the gain comes from the model or from a simple rule; the policy set is organized into baselines, interpretable rules, this paper, and the upper bound.

\TABCAP{Table 9. Admission policies and score sources.}

{\def\LTcaptype{none} 
\begin{longtable}[]{@{}
  >{\raggedright\arraybackslash}p{(\linewidth - 6\tabcolsep) * \real{0.1866}}
  >{\raggedright\arraybackslash}p{(\linewidth - 6\tabcolsep) * \real{0.2742}}
  >{\raggedright\arraybackslash}p{(\linewidth - 6\tabcolsep) * \real{0.3886}}
  >{\raggedright\arraybackslash}p{(\linewidth - 6\tabcolsep) * \real{0.1507}}@{}}
\toprule\noalign{}
\begin{minipage}[b]{\linewidth}\raggedright
Block
\end{minipage} & \begin{minipage}[b]{\linewidth}\raggedright
Policy
\end{minipage} & \begin{minipage}[b]{\linewidth}\raggedright
Basis of admission order
\end{minipage} & \begin{minipage}[b]{\linewidth}\raggedright
Held-out ROC-AUC
\end{minipage} \\
\midrule\noalign{}
\endhead
\bottomrule\noalign{}
\endlastfoot
Baseline & FCFS & Arrival order & --- \\
Baseline & Shortest job first & Declared duration ascending & --- \\
Baseline & Smallest request first & Requested GPUs ascending & --- \\
Baseline & Large request first & Requested GPUs descending & --- \\
Baseline & Random ordering & Uniform random key & --- \\
Interpretable rule & Lookup by GPU model and job type & Historical readiness by model \(\times\) job type, descending & 0.6785 \\
Interpretable rule & Lookup by job type and duration & Historical readiness by job type \(\times\) duration bin, descending & 0.6934 \\
This paper & RAS & Readiness probability of the nine-feature gradient-boosted tree, descending & 0.8176 \\
Upper bound & Ordering by true labels & True readiness label descending & --- \\
\end{longtable}
}

\noindent{\footnotesize Note: all policies share the same placement rule (best-fit) and differ only in admission order; ``held-out'' means the score is estimated outside the training set to which the job belongs.}

Table~9 is organized into four blocks. Random ordering is a lower-bound control carrying no information. The lookup by GPU model and job type has a held-out AUC of 0.6785, almost the same as the 0.6789 reported in Section~4; the lookup by job type and duration (0.6934) lies between it and the three-key joint lookup (0.7520). The last row, ordering by true labels, is used only to indicate the upper bound of the improvement space.
Table~10 gives the evaluation metrics. The four main metrics quantify the same thing from the GPU and job perspectives, and the remaining metrics check whether a policy harms other parts of the system.

\TABCAP{Table 10. Definitions of the evaluation metrics.}

{\def\LTcaptype{none} 
\begin{longtable}[]{@{}
  >{\raggedright\arraybackslash}p{(\linewidth - 6\tabcolsep) * \real{0.2649}}
  >{\raggedright\arraybackslash}p{(\linewidth - 6\tabcolsep) * \real{0.5817}}
  >{\raggedright\arraybackslash}p{(\linewidth - 6\tabcolsep) * \real{0.0855}}
  >{\raggedright\arraybackslash}p{(\linewidth - 6\tabcolsep) * \real{0.0680}}@{}}
\toprule\noalign{}
\begin{minipage}[b]{\linewidth}\raggedright
Metric
\end{minipage} & \begin{minipage}[b]{\linewidth}\raggedright
Definition
\end{minipage} & \begin{minipage}[b]{\linewidth}\raggedright
Direction
\end{minipage} & \begin{minipage}[b]{\linewidth}\raggedright
Role
\end{minipage} \\
\midrule\noalign{}
\endhead
\bottomrule\noalign{}
\endlastfoot
Ineffective-occupancy ratio & GPUs under ineffective occupancy $\div$ total cluster GPUs, i.e., $W(P)$ of Definition~2 in Section~3 & lower & main \\
Effective job ratio & Jobs effectively occupied and placed $\div$ jobs a machine can accommodate & higher & main \\
Effective share among placed jobs & Jobs effectively occupied and placed $\div$ placed jobs & higher & main \\
Large-request effective coverage & Share of requests of 8 GPUs or more that are placed and effectively occupied & higher & main \\
Effective jobs per GPU & Jobs effectively occupied and placed $\div$ placed GPUs & higher & secondary \\
Placeable job ratio & Placed jobs $\div$ jobs a machine can accommodate (placement convention) & --- & control \\
Allocation rate / fairness across request-size bins & Placed GPUs $\div$ total cluster GPUs; Jain index of readiness rates across seven request-size bins & higher & control \\
\end{longtable}
}

\noindent{\footnotesize Note: all metrics refer to a single admission run over the same batch of jobs; ``a machine can accommodate'' in the denominators means that the requested GPUs do not exceed the maximum per-machine GPUs in the cluster.}

\subsection{Experimental Setup and Main Results}

\textbf{Data and protocol:} The demand side comes from the execution-summary table of Section~3.2; the modeling sample over the nine features contains 2,006,496 records (a 5\% row sample of the 40,129,894 labeled records), split into five machine-grouped folds with held-out prediction (the AUC in Table~9 is measured on this split), with an overall readiness rate of 0.4729. The supply side uses three clusters together: two parameterized synthetic heterogeneous clusters (reconstructed from the measured GPUs-per-machine shares and rack homogeneity of Section~4.4, with ``empirical replica'' and ``heavy-tailed'' model mixes, respectively) and a real-snapshot cluster (37,707 machines); demand is aggregated in two ways, record-equal weighting (each record counts once) and weighting by requested GPUs. Load $\rho$ is defined as the ratio of the total requested GPUs submitted in an admission to the total cluster GPUs, with five levels from 1.00 to 2.50; each policy runs 3 seeds per cell (random seeds fixed starting at 20261008), and all policies share the same batch of jobs in the same configuration, so comparisons are paired. All code and data-processing scripts are open source at\hspace{0.35em}\url{https://github.com/liaobin665/gpu-ras}.

Table~11 gives the mean of each policy over all 90 experimental configurations. The four main metrics quantify the same thing from the GPU and job perspectives; we also report the cross-cell standard deviation, the paired difference of each policy against FCFS, and its 95\% confidence interval and paired-test $p$ value. Figure~9 plots the same results in three panels.

\needspace{9\baselineskip}
\begin{samepage}
\TABCAP{Table 11. Main results of the admission policies (3 clusters $\times$ 2 aggregation conventions $\times$ five load levels $\times$ 3 seeds, 90 configurations; each configuration contains 9 policies, 810 runs in total).}
\end{samepage}

{\def\LTcaptype{none} 
\begin{longtable}[]{@{}
  >{\raggedright\arraybackslash}p{(\linewidth - 18\tabcolsep) * \real{0.1953}}
  >{\raggedright\arraybackslash}p{(\linewidth - 18\tabcolsep) * \real{0.0745}}
  >{\raggedright\arraybackslash}p{(\linewidth - 18\tabcolsep) * \real{0.0745}}
  >{\raggedright\arraybackslash}p{(\linewidth - 18\tabcolsep) * \real{0.0745}}
  >{\raggedright\arraybackslash}p{(\linewidth - 18\tabcolsep) * \real{0.0745}}
  >{\raggedright\arraybackslash}p{(\linewidth - 18\tabcolsep) * \real{0.0745}}
  >{\raggedright\arraybackslash}p{(\linewidth - 18\tabcolsep) * \real{0.0745}}
  >{\raggedright\arraybackslash}p{(\linewidth - 18\tabcolsep) * \real{0.0745}}
  >{\raggedright\arraybackslash}p{(\linewidth - 18\tabcolsep) * \real{0.1935}}
  >{\raggedright\arraybackslash}p{(\linewidth - 18\tabcolsep) * \real{0.0897}}@{}}
\toprule\noalign{}
\begin{minipage}[b]{\linewidth}\raggedright
Policy
\end{minipage} & \begin{minipage}[b]{\linewidth}\raggedright
Ineffective-occupancy ratio (\%)
\end{minipage} & \begin{minipage}[b]{\linewidth}\raggedright
Effective job ratio (\%)
\end{minipage} & \begin{minipage}[b]{\linewidth}\raggedright
Effective share among placed jobs (\%)
\end{minipage} & \begin{minipage}[b]{\linewidth}\raggedright
Large-request effective coverage (\%)
\end{minipage} & \begin{minipage}[b]{\linewidth}\raggedright
Allocation rate (\%)
\end{minipage} & \begin{minipage}[b]{\linewidth}\raggedright
Fairness across request-size bins (\%)
\end{minipage} & \begin{minipage}[b]{\linewidth}\raggedright
Cross-cell SD
\end{minipage} & \begin{minipage}[b]{\linewidth}\raggedright
Paired difference vs FCFS (95\% CI, pp)
\end{minipage} & \begin{minipage}[b]{\linewidth}\raggedright
$p$ value
\end{minipage} \\
\midrule\noalign{}
\endhead
\bottomrule\noalign{}
\endlastfoot
FCFS & 43.64 & 34.78 & 48.99 & 46.11 & 98.52 & 89.90 & 5.78 & --- & --- \\
Shortest job first & 47.70 & 34.70 & 49.77 & 40.22 & 98.53 & 81.45 & 4.79 & 4.06 {[}2.75, 5.38{]} & \(3\times 10^{-7}\) \\
Smallest request first & 46.09 & 42.44 & 48.41 & 27.98 & 98.50 & 68.29 & 7.67 & 2.46 {[}1.91, 3.00{]} & \(9\times 10^{-9}\) \\
Large request first & 41.11 & 27.58 & 51.35 & 60.40 & 98.54 & 65.66 & 5.31 & -2.53 {[}-3.09, -1.97{]} & \(6\times 10^{-13}\) \\
Random ordering & 43.48 & 34.87 & 49.05 & 46.07 & 98.52 & 89.71 & 5.75 & -0.16 {[}-0.27, -0.05{]} & \(0.003\) \\
Lookup by model and job type & 34.74 & 34.63 & 59.23 & 59.91 & 98.53 & 82.33 & 7.80 & -8.90 {[}-10.02, -7.77{]} & \(3\times 10^{-15}\) \\
Lookup by job type and duration & 33.88 & 36.04 & 59.69 & 58.01 & 98.53 & 87.49 & 8.37 & -9.75 {[}-11.06, -8.44{]} & \(3\times 10^{-15}\) \\
RAS & 27.15 & 39.48 & 65.85 & 63.18 & 98.53 & 87.26 & 11.02 & -16.49 {[}-18.62, -14.36{]} & \(3\times 10^{-15}\) \\
True-label upper bound & 18.31 & 46.21 & 75.31 & 63.41 & 98.52 & 91.81 & 16.68 & -25.32 {[}-28.87, -21.78{]} & \(3\times 10^{-15}\) \\
\end{longtable}
}

\noindent{\footnotesize Note: all policies use best-fit as the placement rule and differ only in admission order; the first four metrics are the main metrics of Section~5.1 and the last three are controls; ``Cross-cell SD'' is the standard deviation (SD) of each metric across all configurations; the paired difference is ``policy $-$ FCFS'' in percentage points, and the paired test uses $n = 90$ batches under the same experimental configuration.}

\begin{figure*}[t]
\centering
\includegraphics[width=\textwidth]{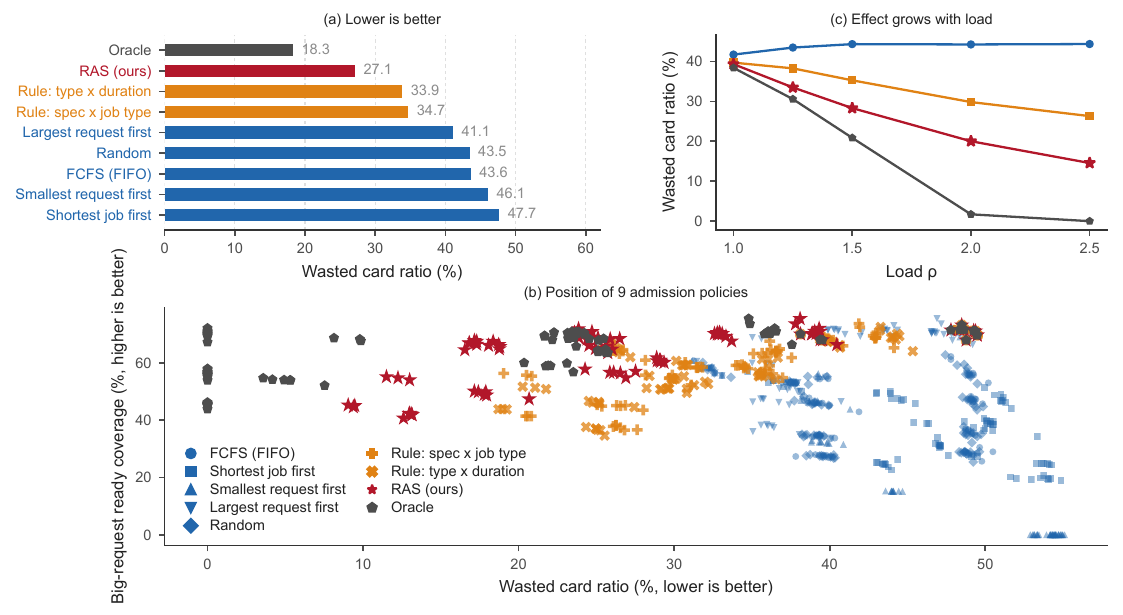}
\caption{(a) Ineffective-occupancy ratio of the nine admission schemes; (b) their positions in the plane of ineffective-occupancy ratio and large-request effective coverage; (c) ineffective-occupancy ratio as a function of load $\rho$.}
\label{fig:9}
\end{figure*}

The main results are threefold.

\textbf{(1) The gain is direct and statistically significant:} Giving admission priority to jobs predicted to become ready directly reduces the GPUs held. RAS lowers the ineffective-occupancy ratio from 43.64\% under FCFS to 27.15\%, a decrease of 16.49 percentage points (95\% CI {[}-18.62, -14.36{]}, $p = 3\times 10^{-15}$); at the same time, the effective share among placed jobs rises from 48.99\% to 65.85\%. That is, for the same set of occupied GPUs, the share actually in use rises markedly, while the total number of occupied GPUs is unchanged (both allocation rates are around 98.53\%). This is not about admitting more jobs but about making placed jobs more likely to actually run. The decrease is statistically significant: the pairing unit is the job batch under the same experimental configuration, and the confidence interval does not cross zero.

\textbf{(2) The gain comes from the model, not from shuffling the order:} The gain depends heavily on the model---simple rules capture only a small fraction of it---and it does not arise from reordering per se. Taking true-label ordering as the theoretical upper bound, RAS captures 65.1\% of the improvement space, whereas the two interpretable rules that use only the contribution-analysis conclusions capture 38.5\% (lookup by job type and duration) and 35.1\% (lookup by GPU model and job type); random ordering yields an ineffective-occupancy ratio of 43.48\%, essentially the same as FCFS's 43.64\%, so without readiness information, changing the order alone brings no gain. This is not in conflict with the contribution analysis: contribution share measures the group-level average contribution of attribute groups to the model's prediction; admission ordering must distinguish, record by record, ``this one will become ready and that one will not'', which relies on the joint discriminative power of the nine attributes. Thus ``three groups dominate'' explains which kinds of attributes the problem concentrates in, but turning it into a usable admission rule requires predicting with all nine attributes jointly, not just the top groups.

\textbf{(3) The gain from admission order grows with load and approaches zero when capacity is ample:} When reading Figure~9(c), focus on the gaps between policies rather than the level of any single curve: the gap between FCFS and RAS grows monotonically from 2.40 percentage points at $\rho = 1.00$ to 29.80 at $\rho = 2.50$ (9.99, 16.02, and 24.24 at $\rho = 1.25$, 1.50, and 2.00, respectively). This is consistent with ``admission order matters only when something must be traded off'': when capacity is sufficient, everything fits regardless of order, so order cannot play a role. The remaining gap to the true-label upper bound is also concentrated at the high-load end: at $\rho = 1.00$, RAS is only 0.92 percentage points from the upper bound, versus 18.31 and 14.56 at $\rho = 2.0$ and 2.5; the remaining gap at $\rho = 2.5$ is smaller than at $\rho = 2.0$ because the true-label upper bound is itself limited by large-request coverage there, and the two policies are nearly tied on the large-request bins.

Two points require explanation. First, the ``effective job ratio'' counts jobs only, and smallest-request-first is higher on it (42.44\%) because it places the largest number of jobs into the cluster; the cost is that it sacrifices almost all large requests, with large-request effective coverage of only 27.98\%, less than half of RAS, and fairness across request-size bins falling from 89.90 to 68.29. ``Highest effective job ratio'' is not the same as optimal; it merely shifts the cost to another class of jobs. RAS has a fairness of 87.26\%, only slightly below FCFS's 89.90\%. Second, the readiness outcomes in this table come from the data; the admission replay experiment changes only the admission order, and the comparison above holds under Assumption~1 of Section~5.1.

\subsection{Ablation and Sensitivity Analyses}

The main results answer whether adjusting the admission order works. This section decomposes the question into four testable ones: does the gain really come from predictor quality; does the conclusion hold under a different placement rule; how does the effect vary with load; and does any request-size bin lose readiness.

\textbf{(1) Predictor-quality sensitivity:} The four policy classes in Table~11 (no information, weak information, joint prediction, information upper bound) themselves constitute this ablation analysis. The ineffective-occupancy ratio decreases along the information order. Random ordering is essentially equivalent to FCFS, showing that the gain is not a side effect of ``shuffling''; the jump from lookup to joint prediction shows that the gain comes from the predictor's ability to use the nine attributes jointly. See Table~11 for the values.

\textbf{(2) Placement-rule sensitivity:} All comparisons above are fixed to best-fit. Replacing the placement rule with first-fit or worst-fit and rerunning the admission schemes gives Table~12.

\TABCAP{Table 12. Placement-rule sensitivity: ineffective-occupancy ratio (\%) under three placement rules.}

{\def\LTcaptype{none} 
\begin{longtable}[]{@{}lllll@{}}
\toprule\noalign{}
Placement rule & FCFS & RAS & True-label upper bound & Gain (pp) \\
\midrule\noalign{}
\endhead
\bottomrule\noalign{}
\endlastfoot
Best-fit & 43.64 & 27.15 & 18.31 & 16.49 \\
First-fit & 43.63 & 27.09 & 18.32 & 16.54 \\
Worst-fit & 40.51 & 29.13 & 23.61 & 11.38 \\
\end{longtable}
}

\noindent{\footnotesize Note: admission order is the same as in Table~11 and only the placement rule changes; ``gain'' is FCFS minus RAS under that placement rule; the three placement rules share the same batch of jobs under the same experimental configuration.}

RAS's gain does not depend on a specific placement rule: under all three rules the difference is positive and in the same direction (16.49, 16.54, and 11.38 percentage points for best-fit, first-fit, and worst-fit); best-fit and first-fit are nearly identical, and the gain shrinks but remains under worst-fit. Worst-fit itself also lowers the ineffective-occupancy ratio, but the gain from changing the admission order persists under all three rules.

\textbf{(3) Load sensitivity:} The conclusion is the same as the third point of Section~5.3.

\textbf{(4) Request-size sensitivity:} Table~13 gives readiness rates for seven request-size bins, to test whether RAS sacrifices any bin.

\needspace{5\baselineskip}
\begin{samepage}
\TABCAP{Table 13. Request-size sensitivity: readiness rates (\%) across seven request-size bins.}
\end{samepage}

{\def\LTcaptype{none} 
\begin{longtable}[]{@{}
  >{\raggedright\arraybackslash}p{(\linewidth - 14\tabcolsep) * \real{0.1250}}
  >{\raggedright\arraybackslash}p{(\linewidth - 14\tabcolsep) * \real{0.1250}}
  >{\raggedright\arraybackslash}p{(\linewidth - 14\tabcolsep) * \real{0.1250}}
  >{\raggedright\arraybackslash}p{(\linewidth - 14\tabcolsep) * \real{0.1250}}
  >{\raggedright\arraybackslash}p{(\linewidth - 14\tabcolsep) * \real{0.1250}}
  >{\raggedright\arraybackslash}p{(\linewidth - 14\tabcolsep) * \real{0.1250}}
  >{\raggedright\arraybackslash}p{(\linewidth - 14\tabcolsep) * \real{0.1250}}
  >{\raggedright\arraybackslash}p{(\linewidth - 14\tabcolsep) * \real{0.1250}}@{}}
\toprule\noalign{}
\begin{minipage}[b]{\linewidth}\raggedright
Policy
\end{minipage} & \begin{minipage}[b]{\linewidth}\raggedright
$\le$ 0.5 GPU
\end{minipage} & \begin{minipage}[b]{\linewidth}\raggedright
(0.5, 1) GPU
\end{minipage} & \begin{minipage}[b]{\linewidth}\raggedright
=1 GPU
\end{minipage} & \begin{minipage}[b]{\linewidth}\raggedright
=2 GPU
\end{minipage} & \begin{minipage}[b]{\linewidth}\raggedright
=4 GPU
\end{minipage} & \begin{minipage}[b]{\linewidth}\raggedright
=8 GPU
\end{minipage} & \begin{minipage}[b]{\linewidth}\raggedright
$\ge$ 16 GPU
\end{minipage} \\
\midrule\noalign{}
\endhead
\bottomrule\noalign{}
\endlastfoot
FCFS & 36.51 & 30.10 & 27.95 & 31.55 & 50.37 & 52.36 & 28.65 \\
RAS & 52.34 & 26.70 & 31.16 & 39.59 & 70.82 & 70.56 & 43.27 \\
\end{longtable}
}

\noindent{\footnotesize Note: readiness rates are record-weighted within each bin; bins are defined by requested GPUs; the two policies share the same batch of jobs under the same experimental configuration.}

Across the seven bins, RAS has a higher readiness rate than FCFS in six, with the largest improvements in the large-request bins (8-GPU: 52.36\% to 70.56\%; 16+ GPU: 28.65\% to 43.27\%), consistent with ``large-request effective coverage from 46.11\% to 63.18\%''. The only exception is the narrow (0.5, 1) GPU bin (30.10\% versus 26.70\%); this bin has the fewest records and carries little weight in the total benefit. The allocation rate of placed GPUs and overall fairness across the seven bins are not harmed (Table~11).

\section{Discussion and Conclusion}

\subsection{Summary}

This paper formalizes ineffective occupancy as a measurable object, measures it at the job level, and proposes a readiness-aware scheduling algorithm on that basis. The main findings are as follows.
Ineffective occupancy is predictable and has a clear structure. Under the machine-grouped five-fold protocol, readiness prediction reaches an ROC-AUC of 0.8187, better than the majority-class baseline (0.4960) and the lookup baselines. The contribution analysis shows that the model's score is driven mainly by job profile, duration, and GPU model, with request size contributing least. Ineffective occupancy is spatially concentrated: measured by the variance ratio $D$ of Eq.~(3), the observed variance of per-machine readiness rates is far larger than the variance expected from a binomial distribution at the same sample size (values of $D$ in Figure~7).
Readiness-aware scheduling substantially reduces ineffective occupancy. With placement fixed and only the admission order adjusted, RAS lowers the ineffective-occupancy ratio from 43.64\% under FCFS to 27.15\% and raises large-request effective coverage from 46.11\% to 63.18\%, without harming capacity utilization. The gain comes mainly from the predictor's discriminative power rather than from a simple order perturbation: compared with the theoretical optimum when every job's true readiness label is known, RAS achieves 65.1\% of the improvement potential, whereas the two interpretable rules derived from the contribution analysis achieve only 38.5\% and 35.1\%.

The work can be summarized in three respects.

\textbf{(1) Existing work focuses on cluster-level allocation and fragmentation rates:} We move the unit of analysis down to individual jobs, reveal the phenomenon of ``GPU capacity is sufficient but the job never becomes ready'', and measure its scale and distribution at the job level.

\textbf{(2) Using only the nine attributes observable at admission time:} We build a readiness predictor and achieve good ranking performance under a machine-grouped five-fold split. The contribution analysis shows that ineffective occupancy is not uniformly distributed random noise but a structural phenomenon organized along job-side attributes such as job profile, duration, and GPU model; it can therefore serve as a stable basis for admission scoring.

\textbf{(3) Given that the capacity constraint on the placement side is not tight and job attributes cannot be changed:} We narrow the adjustable variable to the admission order and propose an admission-ordering rule based on predicted readiness probability. Experiments show that the rule stably reduces ineffective occupancy across clusters and loads, and the allocation rate stays essentially unchanged (98.52\% vs 98.53\%).

In summary, readiness status should become an explicit consideration for heterogeneous GPU cluster schedulers, just like placement success rate.

\subsection{Limitations and Future Work}

The applicability boundaries of our conclusions are as follows.
The readiness label is an observable indicator of readiness status. The release notes define the field as ``whether the pod has ever been ready'', and we use it accordingly as an observable indicator of whether a job can enter the ready state. We also exclude a known alternative explanation: ASI's ``Standby'' mechanism allows users to voluntarily yield compute, a state recorded by the \texttt{state\_public} field of the pod-hour table, which is not in the execution-summary table we use; hence ``not ready'' as defined here does not include voluntary yielding.

Extrapolation requires checking the supply and demand structures. Our conclusions come from observations of a single production cluster. Before extrapolating to other heterogeneous clusters, one should check whether the supply structure (GPUs-per-machine distribution, rack homogeneity) and the demand structure (request-size distribution) are similar.
CPU is not examined separately at the mechanism level. We treat CPU cores per GPU as supply-side structural background (its range is shown in Figure~2) and do not include CPU requests in the model.
The replay evaluation assumes that readiness outcomes do not change with admission order. In reality, image-cache state (warm vs. cold), concurrent startup contention on the same machine, and network congestion can all cause the same job to obtain different readiness outcomes under different admission orders; end-to-end effects require online deployment tests.

Future work can proceed in several directions: first, reproduce our conclusions on other heterogeneous clusters with similar supply and demand structures to test extrapolation; second, include CPU requests in the model and examine the effect of joint CPU--GPU matching on readiness; third, explore finer-grained admission-ordering policies, for example, jointly optimizing predicted readiness, job waiting time, and priority in a multi-objective formulation.

\clearpage
\appendix
\section{Description and Statistical Distributions of the Nine Observable Attributes}

This appendix completes the value ranges and statistical distributions of the nine observable attributes in Table~3 of Section~3.2, for reproduction and verification. Table~A1 gives descriptions and distribution highlights of the nine attributes, and Figure~A1 gives the corresponding statistical distributions.

\TABCAP{Table A1. Description and descriptive statistics of the nine observable attributes.}

{\def\LTcaptype{none} 
\begin{longtable}[]{@{}
  >{\raggedright\arraybackslash}p{(\linewidth - 8\tabcolsep) * \real{0.1750}}
  >{\raggedright\arraybackslash}p{(\linewidth - 8\tabcolsep) * \real{0.1500}}
  >{\raggedright\arraybackslash}p{(\linewidth - 8\tabcolsep) * \real{0.1250}}
  >{\raggedright\arraybackslash}p{(\linewidth - 8\tabcolsep) * \real{0.2800}}
  >{\raggedright\arraybackslash}p{(\linewidth - 8\tabcolsep) * \real{0.2700}}@{}}
\toprule\noalign{}
Feature & Meaning & Values & Distribution highlights & Relation to readiness \\
\midrule\noalign{}
\endhead
\bottomrule\noalign{}
\endlastfoot
\texttt{gpu\_spec\_public} & Accelerator model (request side) & Categorical, 13 classes & Most frequent XPU-E 28.5\%; top three together 57.9\%; smallest class A30 only 0.04\% & Readiness 30.0 (XPU-C) to 58.4 (XPU-E), a range of 28.4 pp \\
\texttt{priority\_class} & Priority class & Categorical, 3 classes & LP 82.8\%, HP 16.1\%, Other 1.1\% & Readiness 44.1 (LP) to 79.9 (Other), a range of 35.8 pp \\
\texttt{job\_type\_public} & Job type & Categorical, 6 classes & offline\_inference 79.7\%, online\_inference 12.3\%, training 7.4\% (remaining 3 classes together 0.6\%) & Readiness 42.4 (offline\_inference) to 91.6 (training), a range of 49.2 pp \\
\texttt{model\_type\_public} & Model type & Categorical, 6 classes & genai 79.4\%, cv 8.9\%, rec 7.0\% (remaining 3 classes together 4.7\%) & Readiness 43.0 (genai) to 86.8 (rec), a range of 43.8 pp \\
\texttt{is\_genai\_request} & Generative-request flag & Binary 0/1 & 77.8\% of records take 1 & Readiness 64.9 at 0 and 42.3 at 1, a difference of 22.6 pp \\
\texttt{gpu\_request} & Requested GPUs & Continuous, including fractional values & Median 1, P90 2; 11.4\% below 1, 0.72\% at 16 or above & Readiness 58.7 at the low end and 56.6 at the high end, a difference of 2.1 pp \\
\texttt{duration\_hours} & Declared duration (hours) & Continuous, heavy-tailed & Median 0.42, mean 5.20, P90 4.65, maximum 1,824 & Readiness 56.1 in the short bin and 64.3 in the long bin, a difference of 8.1 pp \\
\texttt{schedule\_delay\_sec} & Scheduling delay (seconds) & Continuous, more than half zero & Median 0, 55.3\% are 0, P90 224 & Readiness 44.7 at the low end and 53.8 at the high end, a difference of 9.1 pp \\
\texttt{server\_req\_count} & Concurrent requests on the same machine & Continuous & Median 145, mean 192, P90 380 & Readiness 37.6 at the low end and 54.1 at the high end, a range of 16.5 pp \\
\end{longtable}
}

\noindent{\footnotesize Note: statistics are computed on the modeling sample (2,006,496 records from 32,635 machines) and differ from the full-sample conventions of Section~3.2 by less than one percent; ``low'' and ``high'' refer to values at or below the 10th percentile and at or above the 90th percentile, respectively; for categorical attributes only categories with at least 1\% of records are listed; all nine attributes have a missing rate of 0.}

\renewcommand{\thefigure}{A\arabic{figure}}\setcounter{figure}{0}
\begin{figure}[!ht]
\centering
\includegraphics[width=\textwidth]{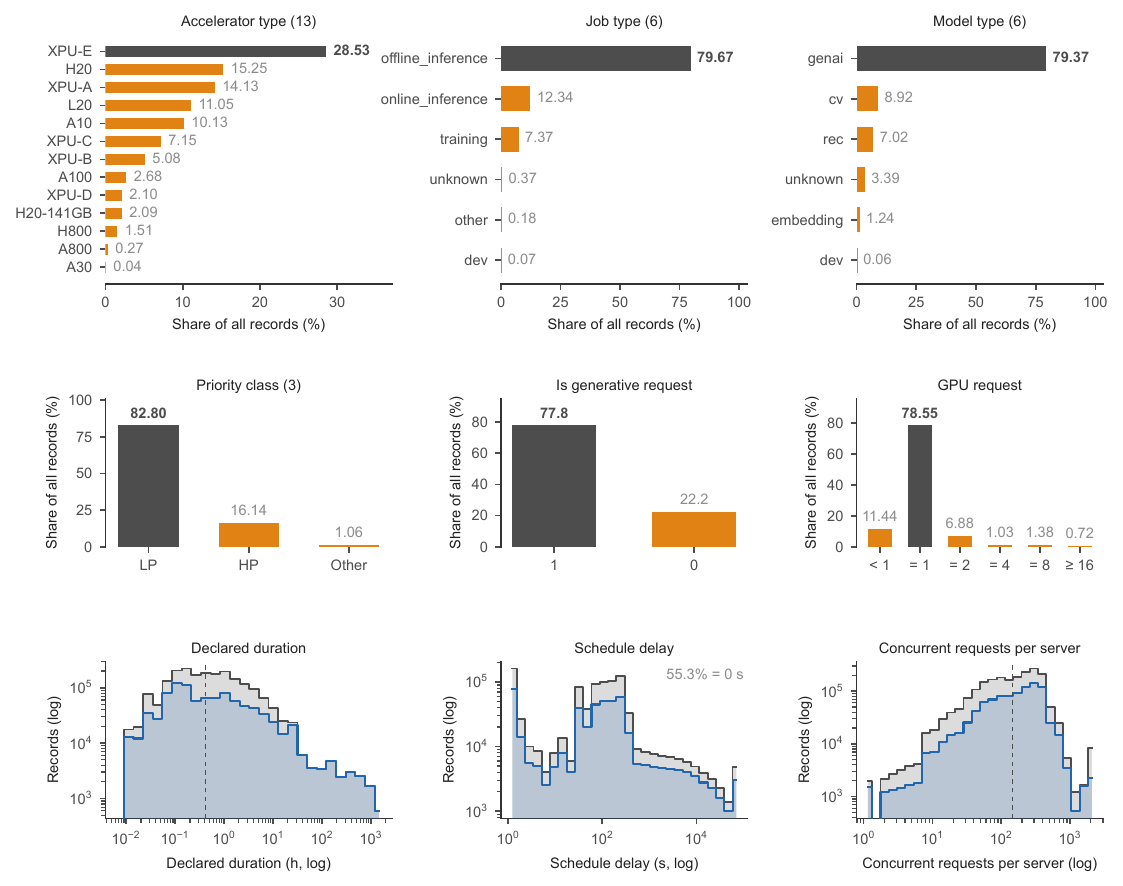}
\caption{Statistical distributions of the nine observable attributes. (a)--(f) show the share of records by value; in (g)--(i), light color denotes all records and dark color the ready records, the horizontal axis is logarithmic, and the dashed line marks the median.}
\label{fig:A1}
\end{figure}

Table~A1 and Figure~A1 together give three facts relevant to the main line of this paper.
First, the demand side is highly concentrated: the distributions of job type, model type, and the generative-request flag are nearly identical: offline inference accounts for about 80\% of records by job type, and generative requests account for a similar share, indicating that the trace is dominated by offline generative inference.
Second, request sizes are dominated by a single GPU, but fractional requests exist (11.4\% below one), so requested GPUs are not strictly integer.
Third, and most relevant here, the difference in readiness rates across the nine attributes is not reflected in request size (a range of only 2.1 percentage points) but is concentrated in job type, model type, and GPU model (ranges of 28.4 to 49.2 percentage points). This is consistent with the feature-group contribution analysis of Section~4.3: what drives whether a job becomes ready is job profile and GPU model, not how many GPUs it requests.

\end{document}